\documentclass[compsoc,conference,a4paper,10pt,times]{IEEEtran}

\usepackage[utf8]{inputenc}
\IEEEoverridecommandlockouts

\usepackage{amsmath,amssymb,amsfonts,amsthm}
\usepackage{algorithmic}
\usepackage{graphicx}
\usepackage{textcomp}
\usepackage{bmpsize}
\usepackage{comment}
\usepackage{xcolor}
\usepackage{enumitem}
\usepackage{footmisc}
\usepackage{biblatex}
\def\BibTeX{{\rm B\kern-.05em{\sc i\kern-.025em b}\kern-.08em
    T\kern-.1667em\lower.7ex\hbox{E}\kern-.125emX}}
\usepackage[ruled,vlined]{algorithm2e}
\usepackage[labelformat=simple]{subcaption}

\usepackage{tikz}
\usetikzlibrary{decorations.pathmorphing}
\usetikzlibrary{arrows}

\newcommand\recht\operatorname
    
\newtheorem{theorem}{Theorem}[section]    
\newtheorem{corollary}[theorem]{Corollary}    

\newtheorem{proposition}[theorem]{Proposition}
\newtheorem{question}[theorem]{Question}
\newtheorem{definition}[theorem]{Definition}
\newtheorem{fact}[theorem]{Fact}
    \newtheoremstyle{TheoremNum}
        {\topsep}{\topsep}              %%% space between body and thm
        {\itshape}                      %%% Thm body font
        {}                              %%% Indent amount (empty = no indent)
        {\bfseries}                     %%% Thm head font
        {.}                             %%% Punctuation after thm head
        { }                             %%% Space after thm head
        {\thmname{#1}\thmnote{ \bfseries #3}}%%% Thm head spec
    \theoremstyle{TheoremNum}

\allowdisplaybreaks

\begin{document}

\title{Improving Frequency Estimation under Local Differential Privacy}

\author{%\IEEEauthorblockN{The authors}
%\IEEEauthorblockA{from somewhere}
\IEEEauthorblockN{Milan Lopuhaä-Zwakenberg\IEEEauthorrefmark{1}, Zitao Li\IEEEauthorrefmark{2}, Boris \v{S}kori\'c\IEEEauthorrefmark{1} and Ninghui Li\IEEEauthorrefmark{2}}
\IEEEauthorblockA{\IEEEauthorrefmark{1}Department of Mathematics and Computer Science \\
Eindhoven University of Technology\\
m.a.lopuhaa@tue.nl, b.skoric@tue.nl}
\IEEEauthorblockA{\IEEEauthorrefmark{2}Department of Computer Sciences \\
Purdue University\\
li2490@purdue.edu, ninghui@cs.purdue.edu}
}

\maketitle

\begin{abstract}
Local Differential Privacy protocols are stochastic protocols used in data aggregation when individual users do not trust the data aggregator with their private data. In such protocols there is a fundamental tradeoff between user privacy and aggregator utility. In the setting of frequency estimation, established bounds on this tradeoff are either nonquantitative, or far from what is known to be attainable. In this paper, we use information-theoretical methods to significantly improve established bounds. We also show that the new bounds are attainable for binary inputs. Furthermore, our methods lead to improved frequency estimators, which we experimentally show to outperform state-of-the-art methods.
\end{abstract}

\section{Introduction}

In a context where a data aggregator collects potentially sensitive data, 
there is an inherent tension between the aggregator's desire to obtain accurate population statistics 
and the individuals' desire to protect their private data. 
One approach to protect privacy is offered by \emph{Local Differential Privacy} (LDP) protocols \cite{kasiviswanathan2011can}. 
Under this framework, each user randomises their private data before sending it to the aggregator. 
This hides the users' true data, while for a large population size the randomness of the users cancels out, 
allowing the aggregator to obtain accurate estimates of the population statistics. 
Because of these properties, LDP mechanisms are widely used in industry by companies such as Apple \cite{apple2017learning}, Google \cite{erlingsson2014rappor}, and Microsoft \cite{ding2017collecting}.

One of the main settings in which LDP protocols are used is that of \emph{frequency estimation} \cite{warner1965randomized,erlingsson2014rappor,wang2017locally}. 
In this setting, 
every user has a private data item from a (finite) set $\mathcal{A}$, 
and the aggregator's goal is to determine the frequencies of the elements of $\mathcal{A}$ among the user population. For example, Chrome uses the LDP protocol RAPPOR to estimate the relative frequencies of homepages and used search engines among its userbase \cite{erlingsson2014rappor}. In this paper, we focus on frequency estimation, as it is both widely studied in the literature, and frequently applied in industry.

In the LDP setting there is a tradeoff between user privacy and frequency estimation accuracy.
Intuitively, the more `random' the privacy protocol is, the better it will hide an individual's private data, 
but the more noisy the aggregator's estimations will be. Therefore, it is natural to ask the following question:

\begin{question} \label{qu:intro}
How can one characterise the relation between user privacy and aggregator utility in the LDP setting?
\end{question}

Of course, to answer this question, we need to choose suitable privacy and utility metrics. As a privacy metric, we focus on $\varepsilon$-LDP, the \emph{de facto} standard privacy metric for stochastic privacy protocols. The privacy parameter $\varepsilon \in \mathbb{R}_{\geq 0}$ provides a measure of the worst-case leakage of the protocol; the advantage of this utility metric is that in practice it is often easily computed for specific protocols, and it provides privacy guarantees that hold in all situations. As a utility metric, we consider the mean squared error (MSE) for the frequency estimation, which is an often-used metric in the literature \cite{duchi2014local,wang2017locally,wang2019locally}.

There are two ways to approach Question \ref{qu:intro}. The first method is to study the tradeoff between privacy and utility for specific protocols, either analytically or experimentally; the second method is to prove theoretic bounds on this tradeoff that hold for any protocol. Unfortunately, there is still a large gap between what can currently be achieved by these two approaches. To be more precise, given $a = \#\mathcal{A}$ categories, $n \gg 0$ users and LDP parameter $\varepsilon \ll 1$, the currently best achievable MSE is $\frac{4a}{n\varepsilon^2}$ \cite{duchi2014local,wang2017locally}, while the theoretical lower bound on the MSE is either $\frac{a}{64n\varepsilon^2}$ or $\Omega(\frac{a}{n\varepsilon^2})$, depending on the precise utility metric (see Section \ref{ssec:contr}). While this shows that asymptotically we can attain the optimal behaviour in $a,n,\varepsilon$ up to a constant, in practice one needs to know this constant if one is to determine a satisfactory level of privacy based on $a$, $n$, and utility demands. Therefore, it is important to bridge the gap between practice and theoretical lower bounds.

One of the factors that plays a role in the existence of this gap is that a privacy protocol typically consists of two algorithms $(\mathcal{Q},\Phi)$, where $\mathcal{Q}$ is employed by the user to randomise their data, and $\Phi$ is used by the aggregator to obtain frequency estimations from the randomised data. However, it is unknown what the optimal $\Phi$ is given $\mathcal{Q}$ \cite{wang2019locally}. Typically, one first uses an unbiased, linear estimator $\Phi'$ called a frequency oracle \cite{warner1965randomized,erlingsson2014rappor,wang2017locally}, and then postprocesses the results to make them more in line with what the aggregator expects a frequency distribution to look like \cite{jia2019calibrate,wang2019locally}. However, theoretical analyses of the utility of postprocessing methods are lacking. Furthermore, the estimator $\Phi'$ typically discards at least some of the information present in the randomised data, which leaves open the possibility that postprocessing the outcome of $\Phi'$ does not lead to the optimal $\Phi$.

%---------------------------------------------------------------
\subsection{Our contributions} \label{ssec:contr}

This paper presents two contributions towards answering Question \ref{qu:intro}, by working towards closing the gap in the privacy-utility tradeoff between existing protocols and theoretical lower bounds. To outline these contributions we first need to introduce some notation. In this paper, we consider two different, but related estimation problems: (i) the aggregator wants to estimate the probability distribution $P$ from which the users' private data is drawn, and (ii) the aggregator wants to estimate the actual frequencies $F$ of the private data. The distribution $P$ is unknown to the aggregator, and as such we can consider it to be a random variable itself; its distribution reflects the aggregator's prior knowledge. It turns out that for a given protocol $\mathcal{Q}$, the optimal estimators $\Pi$ for $P$ and $\Phi$ for $F$ can be stated explicitely in terms of this prior distribution:

\begin{theorem}[Informal version of Theorem \ref{thm:bound1}]
Let $\mathcal{Q}$ be a privacy protocol, and for each $i \leq n$, let $Y_i := \mathcal{Q}(X_i)$ be the randomisation of user $i$'s private data. Then given $\vec{Y} = \vec{y}$, the optimal estimators $\Pi_{\recht{opt}}$ for $P$ and $\Phi_{\recht{opt}}$ for $F$ are given by
\begin{align}
\Pi_{\recht{opt}}(\vec{y}) &= \mathbb{E}[P|\vec{Y}= \vec{y}],\\
\Phi_{\recht{opt}}(\vec{y}) &= \mathbb{E}[F|\vec{Y}= \vec{y}].
\end{align}
\end{theorem}

This theorem `solves'  the problem of postprocessing. Unfortunately, computing $\Pi_{\recht{opt}}$ and $\Phi_{\recht{opt}}$ can be time-consuming in practice: if $b$ is the size of the output space of $\mathcal{Q}$, then the time complexity is $\mathcal{O}(n^{a(b-1)})$ under moderate assumptions on $\mu$, which grows unfeasibly large for large $n$. However, the following theorem shows that we can approximate $\Pi_{\recht{opt}}$ and $\Phi_{\recht{opt}}$ by estimators which are much easier to compute.

\begin{theorem}[Informal version of Theorem \ref{thm:mle}]
Let $\Pi_{\recht{MLE}}(\vec{y})$ be the maximum likelihood estimator of $P$ given $\vec{y}$. Then $\Pi_{\recht{MLE}} \rightarrow \Pi_{\recht{opt}}$ in probability as $n \rightarrow \infty$. Furthermore, one finds $\Pi_{\recht{MLE}}(\vec{y})$ by solving a convex optimisation problem of dimension $(a-1)$ of which the complexity does not depend on $n$.
\end{theorem} 

Furthermore, we perform experiments that show that $\Pi_{\recht{MLE}}$ outperforms state-of-the-art postprocessing methods, demonstrating the validity of our approach. Since $F \approx P$ for large $n$, we can also use $\Pi_{\recht{MLE}}$ to approximate $\Phi_{\recht{opt}}$. This reduces the problem of finding the optimal $(\mathcal{Q},\Pi)$ or $(\mathcal{Q},\Phi)$ to finding the optimal $\mathcal{Q}$.

On the other hand, we use the  information-theoretic methods from \cite{lopuhaa2019information} to find new lower bounds:

\begin{theorem}[Informal version of Theorem \ref{cor:eps}]
For $n \gg 0$ and $\varepsilon \ll 1$, one has, for any $\mathcal{Q}$, $\Pi$ and $\Phi$, and any distribution for~$P$:
\begin{align*}
\recht{MSE}^{\recht{distr}}(\mathcal{Q},\Pi), \recht{MSE}^{\recht{freq}}(\mathcal{Q},\Phi)  &\geq \frac{a}{n\varepsilon^2},
\end{align*}
where $\recht{MSE}^{\recht{distr}}(\mathcal{Q},\Pi)$ is the expected mean squared error for distribution estimation, over the distributions of $P$ and the stochastic function $\mathcal{Q}$, and $\recht{MSE}^{\recht{freq}}(\mathcal{Q},\Phi)$ is the mean squared error for frequency estimation.
\end{theorem}

This result significantly improves known lower bounds, as can be seen from Table \ref{tab:results}. Not only does this provide us with a constant for $\recht{MSE}^{\recht{freq}}(\mathcal{Q},\Phi)$ where first none was known, and significantly improves the constant for $\recht{MSE}^{\recht{distr}}(\mathcal{Q},\Pi)$, it also holds for any prior distribution, rather than just the worst-case one. The downside, however, is that it only holds for asymptotically large $n$. However, one of the settings in which LDP is typically employed is one where $n \gg a$. It should be noted that in this introduction we write the results in terms of the limit case $\varepsilon \rightarrow 0$ for simplicity, but our results also improve known bounds for any $\varepsilon$ (see Section \ref{sec:bounds}).

{\renewcommand{\arraystretch}{1.3}
\begin{table}
\centering
\caption{New results compared to known lower bounds and attainable MSE for $\varepsilon \ll 1$.} \label{tab:results}
\begin{tabular}{|lcl|}
\hline
\multicolumn{3}{|c|}{Distribution estimation} \\ \hline
Known lower bound \cite{duchi2013local} & $\frac{a}{64n\varepsilon^2}$ &  worst case \\
Attainable \cite{acharya2019hadamard,duchi2013local} & $\frac{4a}{n\varepsilon^2}$ & \\
New lower bound & $\frac{a}{n\varepsilon^2}$ & any case, for $n \rightarrow \infty$ \\
\hline
\hline
\multicolumn{3}{|c|}{Frequency estimation} \\ \hline
Known lower bound \cite{blasiok2019towards} & $\Omega(\frac{a}{n\varepsilon^2})$ & worst case\\
Attainable \cite{wang2017locally} & $\frac{4a}{n\varepsilon^2}$ &\\
New lower bound & $\frac{a}{n\varepsilon^2}$ & any case, for $n \rightarrow \infty$ \\ \hline
\end{tabular}

\end{table}
}

While earlier works also use information-theoretic methods to derive lower bounds, our methods rely on and expand upon the description of the asymptotic behaviour of the mutual information $\recht{I}(P;\vec{Y})$ in \cite{lopuhaa2019information}. Since this description gives a limit, rather than a lower bound, this allows us to be more precise than earlier work.

The structure of this paper is as follows. In Section \ref{sec:prelim} we introduce the mathematical setting for this paper. In Section \ref{sec:opt}, we discuss how to find, and approximate computationally, the optimal $\Pi$ given $\mathcal{Q}$. In Section \ref{sec:bounds}, we prove the new lower bounds for the MSE. In Section \ref{sec:exp}, we evaluate the methods from Section \ref{sec:opt} experimentally.

\subsection{Related work}

A lower bound for distribution estimation is given in Proposition 6 of \cite{duchi2014local}; from the proof in the arXiv version \cite{duchi2013local} we find that we can express this lower bound as $\frac{a}{64n\varepsilon^2}$ for $\varepsilon \ll 1$. Their strategy is to first relate frequency estimation to a binary hypothesis testing problem, and then use information theory to bound the accuracy of this testing problem. They also show that a frequency MSE of $\frac{4a}{n\varepsilon^2}$ can be attained by providing a specific protocol, namely a version of the privacy protocol RAPPOR \cite{erlingsson2014rappor}. Another privacy protocol with the same frequency MSE, inspired by coding theory, is described in \cite{acharya2019hadamard}.

The work in \cite{blasiok2019towards} looks at the `dual' problem, i.e. given $\varepsilon$, $a$, and an acceptable frequency MSE\footnote{In the notation of \cite{blasiok2019towards} we would have $\kappa = \alpha^2$.} $\kappa$, what is the minimal $n$ for which a protocol exists that satisfies these criteria? In Theorem 1.7, they find that $n \geq \Omega(\frac{a}{\kappa\varepsilon^2})$ (even when allowing for protocols that only offer $(\varepsilon,\delta)$-LDP, a weaker privacy notion), which we can restate as $\kappa \geq \Omega(\frac{a}{n\varepsilon^2})$. Their result is based on techniques in \cite{bassily2015local} to bound the mutual information between the input and output of $\mathcal{Q}$ in terms of the privacy parameters $(\varepsilon,\delta)$.

In \cite{wang2017locally}, it is proven that RAPPOR attains a frequency MSE of $\frac{4a}{n\varepsilon^2}$. It is also shown there that this can be improved upon by tweaking its parameters, leading to a protocol called Optimised Unary Encoding (UE). This does not affect its asymptotic behaviour for $\varepsilon \ll 1$.

Instead of looking at the MSE, which is essentially the $\ell_2$-distance between the true value and its estimation, one can also consider other error metrics, such as $\ell_1$ or $\ell_{\infty}$. An overview of the behaviour of these error metrics is found in \cite{cheu2019manipulation}.

There is much literature on studying LDP via information theory, and using this to derive properties on privacy or utility \cite{bassily2015local,cuff2016differential,duchi2014local,kairouz2014extremal}. We mostly rely on the techniques of \cite{lopuhaa2019information}, where information-theoretic properties are stated as (asymptotic) equalities, rather than inequalities; this allows us to obtain tighter bounds than earlier work.

There exists a body of work on post-processing LDP results \cite{cormode2019answering,jia2019calibrate,wang2018privtrie,wang2019locally}. Heuristically, these approaches are based on the idea that one can improve an estimation by taking advantage of the knowledge one has of what a distribution should look like. In \cite{jia2019calibrate}, the estimation is adjusted to adhere more to a power law-type distribution; in \cite{wang2019locally}, the estimation is adjusted to account for the fact that the tallies should be nonnegative and add up to $1$. The validity of these approaches is supported by empirical evidence, but a theoretical analysis is lacking. In this paper, we formalise the intuition that frequency estimation can be improved by taking prior knowledge into account, leading to the optimal $\Pi$ and $\Phi$ given $\mathcal{Q}$.

\section{Preliminaries} \label{sec:prelim}

In this section, we introduce our setting and various concepts that play a role in this paper.

\subsection{Setting}

We consider the setting of \emph{Local Differential Privacy} (LDP). There are $n$ users, and user $i$ has a private data item $X_i$ from a finite set $\mathcal{A}$. Let $\mathcal{P}_{\mathcal{A}}$ be the space of probability distributions on $\mathcal{A}$; then we assume that there is a distribution $P \in \mathcal{P}_{\mathcal{A}}$ such that each $X_i$ is drawn independently from $P$. The distribution $P$ itself is unknown to the aggregator, so we consider it to be a random variable itself, taken from a continuous prior distribution $\mu$ on $\mathcal{P}_{\mathcal{A}}$. The distribution $\mu$ is known to the aggregator, and reflects their prior knowledge.  The aggregator publishes a \emph{privacy protocol}, i.e. a random function $\mathcal{Q}\colon \mathcal{A} \rightarrow \mathcal{B}$. User $i$ calculates $Y_i := \mathcal{Q}(X_i)$ and sends it to the aggregator. The aggregator is interested in one of two quantities unavailable to them:

 \begin{enumerate}
\item The aggregator wants to know $P$ as accurately as possible. This occurs, for instance, when the aggregator is a scientist whose userbase is a sample of a greater population \cite{warner1965randomized}. 
The aggregator is not concerned with this specific userbase, but rather with the characteristics of the general population, which are modelled by $P$.
\item For $a=  \#\mathcal{A}$, define the frequency vector $F \in \mathbb{R}_{\geq 0}^{a}$ by $F_x = \tfrac{\#\{i: X_i = x\}}{n}$ for all $x \in \mathcal{A}$. Then the aggregator wants to know $F$ as accurately as possible. This occurs, for instance, when the users are customers of a service, and the service provider wants to know statistics about is customer base \cite{erlingsson2014rappor}.
\end{enumerate}

Note that for large $n$, one has $F \approx P$, so these goals are closely aligned in practice. Nevertheless, we make a distinction between these two cases, as the means of obtaining their lower bounds are different. 
To estimate $P$, the aggregator employs a distribution estimator $\Pi$, which takes as input the perturbed data $\vec{Y}$, and outputs an estimation $\hat{P} = \Pi(\vec{Y}) \in \mathbb{R}^{a}$ of $P$. To estimate $F$, the aggregator likewise computes $\hat{F} = \Phi(\vec{Y}) \in \mathbb{R}^{a}$. 
The setting is depicted in Figure \ref{fig:ldp}, and our notation is listed in Table \ref{tab:notation}; some notation will be defined later in this section. 
In particular, for $\alpha \in \mathcal{A}$, and $\beta \in \mathcal{B}$, we will make use of the notations
\begin{align}
T_\alpha &= \#\{i:X_i = \alpha\} = nF_{\alpha}, \label{eq:deft}\\
S_\beta &= \#\{i:Y_i = \beta\}, \label{eq:defs}\\
S_{\beta|\alpha} &= \#\{i: X_i = \alpha, Y_i = \beta\}. \label{eq:defs2}
\end{align}
We write $T$ and $S$ for the vectors in $\mathbb{Z}^{a}_{\geq 0}$ and $\mathbb{Z}^{b}_{\geq 0}$, where $b = \#\mathcal{B}$. For a particular $\vec{x} \in \mathcal{A}^n$, $\vec{y} \in \mathcal{B}^n$, we write $t_\alpha$, $s_{\beta}$, $s_{\beta|\alpha}$ for the associated tallies; we write $t(\vec{x})$, etc., if we need to disambiguate between different input data.

\begin{figure*}
\begin{center}
\begin{tikzpicture}[scale=0.4]

\draw[rounded corners] (7,1) -- (11,1) -- (11,-7) -- (7,-7) -- cycle;
\draw[-] (9,1) -- (9,-7);
\draw[-] (9,-1) -- (11,-1);
\draw[-] (9,-3) -- (11,-3);
\draw[-] (9,-5) -- (11,-5);
\draw (8,-3) node {$\vec{X}$};
\draw (10,0) node(a1) {$X_1$};

\draw[-latex,decorate,decoration={snake}] (11,0) --node[above]{$\mathcal{Q}$} (15,0);
\draw (10,-2) node(a2) {$X_2$};

\draw[-latex,decorate,decoration={snake}] (11,-2) --node[above]{$\mathcal{Q}$} (15,-2);
\draw (10,-6) node(an) {$X_n$};

\draw[-latex,decorate,decoration={snake}] (11,-6) --node[above]{$\mathcal{Q}$} (15,-6);
\draw[-, dashed] (10,-3.5) -- (10,-4.5);

\draw[rounded corners] (15,1) -- (19,1) -- (19,-7) -- (15,-7) -- cycle;
\draw[-] (17,1) -- (17,-7);
\draw[-] (15,-1) -- (17,-1);
\draw[-] (15,-3) -- (17,-3);
\draw[-] (15,-5) -- (17,-5);
\draw (16,0) node {$Y_1$};
\draw (16,-2) node {$Y_2$};
\draw (16,-6) node {$Y_n$};
\draw[-, dashed] (16,-3.5) -- (16,-4.5);
\draw (18,-3) node {$\vec{Y}$};
\draw[dashed, rounded corners] (-0.5,1.5) --node[above]{Hidden} (11.5,1.5) -- (11.5,-7.5) -- (-0.5,-7.5) -- cycle;
\draw[rounded corners] (0,0.5) -- (3,0.5) -- (3,-2.5) -- (0,-2.5) -- cycle;
\draw (1.5,-1) node {$P$};
\draw[-latex,decorate,decoration={snake}] (3,-1) --node[above]{$X_i \sim P$} (7,-1);
\draw[rounded corners] (0,-3.5) -- (3,-3.5) -- (3,-6.5) -- (0,-6.5) -- cycle;
\draw (1.5,-5) node{$F$};
\draw[-latex] (7,-5) --node[above]{frequencies} (3,-5);
\draw (-5.5,-1) node {$\mu$};
\draw[rounded corners] (-7,0.5) -- (-4,0.5) -- (-4,-2.5) -- (-7,-2.5) -- cycle;
\draw[-latex,decorate,decoration={snake}] (-4,-1) --node[above]{$P \sim \mu$} (0,-1);
\draw[rounded corners] (22,0.5) -- (25,0.5) -- (25,-2.5) -- (22,-2.5) -- cycle;
\draw[rounded corners] (22,-3.5) -- (25,-3.5) -- (25,-6.5) -- (22,-6.5) -- cycle;
\draw[-latex] (19,-1) --node[above]{$\Pi$} (22,-1);
\draw[-latex] (19,-5) --node[above]{$\Phi$} (22,-5);
\draw (23.5,-1) node{$\hat{P}$};
\draw (23.5,-5) node{$\hat{F}$};
\draw (-5.5,-9) node {prior};
\draw (1.5,-9) node {statistics};
\draw (9,-9) node {private data};
\draw (17,-9) node {perturbed data};
\draw (23.5,-9) node {estimates};
\end{tikzpicture}
\end{center}
\caption{The LDP setting as used in this paper. \label{fig:ldp}}
\end{figure*}
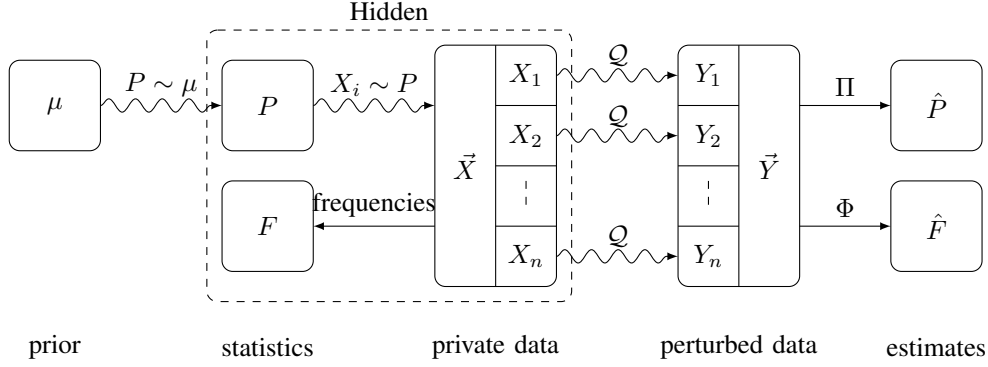

\begin{table}
    \centering
    \begin{tabular}{|c|c|}
    \hline
        $\mathcal{A}$ & input space \\
        $a$ & $\#\mathcal{A}$ \\
        $n$ & number of users\\
        $\vec{X}$ & vector of private data\\
        $F$ & frequencies of private data \\
        $T$ & tallies of private data \\
        \hline
        $\mathcal{P}_{\mathcal{A}}$ & space of prop. distr. on $\mathcal{A}$ \\
        $P$ & prob. distr. of private data \\
        $\mu$ & probability measure of $P$ \\
        \hline
        $\mathcal{Q}$ & privacy protocol \\
        $Q$ & matrix associated to $\mathcal{Q}$ \\
        $\mathcal{B}$ & output space associated to $\mathcal{Q}$ \\
        $b$ & $\#\mathcal{B}$\\
        $\varepsilon$ & LDP-parameter \\
        \hline
        $\vec{Y}$ & vector of perturbed data\\
        $S$ & tallies of perturbed data\\
        $\Pi$ & estimator of $P$ given $\vec{Y}$\\
        $\Phi$ & estimator of $F$ given $\vec{Y}$\\
        \hline
    \end{tabular}
    \caption{\it Notation employed in this paper.}
    \label{tab:notation}
\end{table}

\subsection{Privacy and utility metrics}

Let $\mathcal{A}$ be a finite set. A \emph{privacy protocol} for $\mathcal{A}$ is a random function $\mathcal{Q}\colon\mathcal{A} \rightarrow \mathcal{B}$, where $\mathcal{B}$ is another finite set. 
Upon identifying $\mathcal{A} = \{1,\cdots,a\}$ and $\mathcal{B} = \{1,\cdots,b\}$, we can represent $\mathcal{Q}$ by a matrix $Q \in \mathbb{R}^{b \times a}$, with $Q_{y|x} = \mathbb{P}(\mathcal{Q}(x) = y)$.  The \emph{de facto} standard way to measure the privacy of a protocol is via $\varepsilon$-Local Differential Privacy (LDP):

\begin{definition}
\emph{($\varepsilon$-LDP \cite{kasiviswanathan2011can})}
Let $\mathcal{Q}$ be a privacy protocol for $\mathcal{A}$, and let $\varepsilon \in \mathbb{R}_{\geq 0}$. We say that $\mathcal{Q}$ \emph{satisfies $\varepsilon$-LDP} if for all $x,x' \in \mathcal{A}$ and all $y \in \mathcal{B}$ one has
\begin{equation}
Q_{y|x'} \leq \textrm{\emph{e}}^{\varepsilon}Q_{y|x}.
\end{equation}
\end{definition}

Intuitively, this means that for small $\varepsilon$, given the output $y$, it is difficult to decide whether the input was $x$ or $x'$. The smaller $\varepsilon$, the more privacy the protocol offers. An advantage of this metric is that it is a worst-case approach to privacy, ensuring strong privacy guarantees that hold in all situations.

On the side of utility, we measure the accuracy of the aggregator's estimation by taking the expected value (over $P \sim \mu$) of the mean squared error:

\begin{definition}
We define the \emph{mean squared error} of $(\mathcal{Q},\Pi)$ and $(\mathcal{Q},\Phi)$ to be
\begin{align}
\recht{MSE}_{\mu}^{\recht{distr}}(\mathcal{Q},\Pi) &:= \mathbb{E}_{P,\vec{Y}}||P-\Pi(\vec{Y})||_2^2, \label{eq:msedef1}\\
\recht{MSE}_{\mu}^{\recht{freq}}(\mathcal{Q},\Phi) &:= \mathbb{E}_{F,\vec{Y}}||F-\Phi(\vec{Y})||_2^2. \label{eq:msedef2}
\end{align}
\end{definition}

Note that these metrics depend on $a$ and $n$. They also depend on $\mu$, as this affects the distributions of both $P$ and $F$. Using this metric means that the best estimator is the one that gives on average the lowest squared error, when averaging over all possible input distributions, and all possible outputs. If one is more interested in worst-case performance, one can consider the following worst-case metrics \cite{duchi2013local,blasiok2019towards}:
\begin{definition}
We define the \emph{worst-case MSE} of $(\mathcal{Q},\Pi)$ and $(\mathcal{Q},\Phi)$ to be

\begin{align}
\recht{WMSE}^{\recht{distr}}(\mathcal{Q},\Pi) &:= \sup_{p}\mathbb{E}_{\vec{Y}|P=p}\left[||p-\Pi_n(\vec{Y})||^2_2\right], \\
\recht{WMSE}^{\recht{freq}}(Q,\Phi_n) &:= \max_f \mathbb{E}_{\vec{Y}|F=f}\left[||f-\Phi_n(\vec{Y})||^2_2\right]. 
\end{align}
\end{definition}
These are the metrics used by \cite{blasiok2019towards,duchi2013local} in Table \ref{tab:results}. Note that
\begin{align}
\recht{WMSE}^{\recht{distr}}(\mathcal{Q},\Pi) &= \sup_{\mu} \recht{MSE}^{\recht{distr}}_{\mu}(\mathcal{Q},\Pi),\\
\recht{WMSE}^{\recht{freq}}(\mathcal{Q},\Phi) &\geq \sup_{\mu} \recht{MSE}^{\recht{freq}}_{\mu}(\mathcal{Q},\Phi).
\end{align}
We prefer to use the MSE measures 
instead of the worst-case measures
for two reasons: First, they align more closely with an aggregator's needs, since the aggregator will be interested in a protocol's behaviour on a typical database. 
Second, while worst case guarantees are helpful, the WMSE metrics do not fully give worst case guarantees, as they still involve an expected value over $\mathcal{Q}$. In fact, since the LDP condition implies that any output has positive probability given any input, it is impossible to give a worst case guarantee on the estimation error. Therefore it makes sense to consider average case metrics. Third, any lower bound on MSE will imply a lower bound on WMSE, making it more fruitful to look for lower bounds on MSE.

\subsection{Probability and information theory} \label{ssec:prob}

For a finite set $\mathcal{A}$ of size $a$, we write $\mathcal{P}_{\mathcal{A}}$ for the space of probability distributions on $\mathcal{A}$, i.e.
\begin{equation}
\mathcal{P}_{\mathcal{A}} = \left\{p \in \mathbb{R}^{a}_{\geq 0} : \sum_{x \in \mathcal{A}} p_x = 1\right\}.
\end{equation}
For a discrete random variable $X$ on $\mathcal{A}$, we write $\recht{H}(X)$ for its Shannon entropy, and for a continuous random variable $Y$ on $\mathbb{R}^d$, we write $\recht{h}(Y)$ for its differential entropy. If $P$ is a continuous random variable on $\mathcal{P}_{\mathcal{A}}$, we define its differential entropy as follows: Choose an identification $\mathcal{A} = \{1,\cdots,a\}$, and define $P' = (P_1,\cdots,P_{a-1})$, which is a continuous random variable on $\mathbb{R}^{a-1}$. We then define $\recht{h}(P) := \recht{h}(P')$; this does not depend on the choice of enumeration of $\mathcal{A}$. As such, we fix such an enumeration for the rest of this paper, unless stated otherwise. Related information-theoretic measures, such as $\recht{h}(P|\vec{Y})$ and $\recht{I}(P;\vec{Y})$, are defined similarly. For more details on information theory we refer to \cite{cover1999elements}.

\subsection{Assumptions on privacy protocols} \label{ssec:ass}

Throughout this paper we make two technical but harmless assumptions on the privacy protocol~$\mathcal{Q}$:
\begin{enumerate}
    \item We assume that the matrix $Q$ has rank $a$, i.e. it is injective as a linear map. We do this because privacy protocols of lower rank are unable to distinguish all possible input distributions, which makes them unsuitable for frequency estimation \cite{lopuhaa2019information}.
    \item  We assume that for all $x \in \mathcal{A}, y \in \mathcal{B}$ one has $Q_{y|x} > 0$. If there is an $y$ such that $Q_{y|x} = 0$ for all $x$, then we can remove $y$ from $\mathcal{B}$ without changing the protocol. If there is an $y$ such that $Q_{y|x} = 0$ but $Q_{y|x'} > 0$, then $Q$ does not satisfy $\varepsilon$-LDP for any $\varepsilon$. Since we are only interested in privacy protocols that offer privacy, disregarding this case is harmless.
\end{enumerate}
The main reason for these assumptions is that they simplify the expressions in Theorem \ref{thm:hlim} compared to \cite{lopuhaa2019information}, while still being general enough to describe all protocols that are used in practice.

%============================
\section{Optimal estimators} \label{sec:opt}

We obtain our first upper bound for the MSE by giving formulas for the optimal estimators, given $Q,\mu$. Essentially, this is a well-known fact about minimum mean square error (MMSE) estimators. Recall that we have identified $\mathcal{A} = \{1,\cdots,a\}$.

\begin{theorem} \label{thm:bound1}
Let $\mathcal{Q},\mu$ be given. The $\Pi_{\recht{opt}}$ and $\Phi_{\recht{opt}}$ minimising (\ref{eq:msedef1}) and (\ref{eq:msedef2}), respectively, are given by
\begin{align}
\Pi_{\recht{opt}}(\vec{y}) &= \mathbb{E}_{P|\vec{Y}=\vec{y}}[P],\\
\Phi_{\recht{opt}}(\vec{y}) &= \mathbb{E}_{F|\vec{Y}=\vec{y}}[F].
\end{align}
For these estimators we have
\begin{align}
\recht{MSE}^{\recht{distr}}_{\mu}(\mathcal{Q},\Pi_{\recht{opt}}) &= \sum_{x=1}^a \mathbb{E}_{\vec{y}}\recht{Var}(P_x|\vec{Y}=\vec{y}),\\
\recht{MSE}^{\recht{freq}}_{\mu}(\mathcal{Q},\Phi_{\recht{opt}}) &= \sum_{x=1}^a \mathbb{E}_{\vec{y}}\recht{Var}(F_x|\vec{Y}=\vec{y}).
\end{align}
\end{theorem}

\begin{proof}
For any estimator $\Pi$ of $P$ one has
\begin{align}
\recht{MSE}_{\mu}^{\recht{distr}}(Q,\Pi) &= \mathbb{E}_{\vec{y}}\mathbb{E}_{P|\vec{Y}=\vec{y}}\left[||P-\Pi(\vec{y})||^2_2\right]\\
&= \mathbb{E}_{\vec{y}}\sum_{x=1}^a\mathbb{E}_{P|\vec{Y}=\vec{y}}\left[(P_x-\Pi(\vec{y})_x)^2\right]\\
&\geq \mathbb{E}_{\vec{y}}\sum_{x=1}^a\recht{Var}(P_x|\vec{Y}=\vec{y}). \label{eq:varbound}
\end{align}
One has equality in (\ref{eq:varbound}) if and only if $\Pi(\vec{y})_x = \mathbb{E}_P[P_x|\vec{Y}=\vec{y}]$ for all $x$ and $\vec{y}$. The proof for $\Phi$ is similar.
\end{proof}

Theorem \ref{thm:bound1} formalises the intuition in \cite{jia2019calibrate,wang2019locally} that one gets better estimators by incorporating prior knowledge of the distribution. While Theorem \ref{thm:bound1} gives us a direct formula for the optimal estimators for a given privacy protocol, the disadvantage is that these can be computationally difficult to evaluate. With regards to $\Pi_{\recht{opt}}$, we find that for any $x$ we have
\begin{equation}
[\Pi_{\recht{opt}}(\vec{y})]_{x} = \frac{1}{\mathbb{P}(\vec{Y}=\vec{y})}\int_{p \in \mathcal{P}_{\mathcal{A}}} \mu(p)p_{x} \prod_{\beta \in \mathcal{B}} (Q\cdot p)_\beta^{s_\beta} \textrm{d}p, \label{eq:int}
\end{equation}
where $Q \cdot p$ is matrix multiplication, $s_\beta$ is as in (\ref{eq:defs}), and $\mu$ is the probability density function for $p$. For large $a$, the integral over $\mathcal{P}_{\mathcal{A}}$ can be computationally involved. One approach is to do a Monte Carlo approximation of the integral. Typically, $\mu$ will be a Dirichlet distribution, for which several efficient Monte Carlo methods exist \cite{cheng1998random,betancourt2012cruising}. However, for large $n$, the $s_\beta$ will be large as well. This leads to a spiky distribution, which needs more samples to approximate accurately. If the aggregator is interested in estimating $P$, we can circumvent this by giving an expression for $\Pi_{\recht{opt}}(\vec{y})$, as well as for its MSE that does not involve any integration. Its complexity is stated in the following Proposition:

\begin{proposition}
\label{prop:dirichlet}
Let $\Pi_{\recht{opt}}$ be as in Theorem \ref{thm:bound1}. Suppose $\mu$ is a Dirichlet distribution, and let $\vec{y} \in \mathcal{B}^n$. Then the estimator $\Pi_{\recht{opt}}(\vec{y})$ and $\recht{MSE}^{\recht{distr}}_{\mu}(\mathcal{Q},\Pi_{\recht{opt}})$ can be calculated in time complexity $\mathcal{O}(n^{b(a-1)})$ and $\mathcal{O}(n^{ab-1})$, respectively.
\end{proposition}

This Proposition is proven in appendix \ref{app:dirichlet}. Unfortunately, since $n$ is typically very large in LDP settings, this can get computationally prohibitive. It therefore becomes useful to look for ways to approximate $\Pi_{\recht{opt}}$.
Below we discuss such an approximation method.

\subsection{Approximation by MLE} \label{ss:mle}

For large $n$, the calculation in Proposition \ref{prop:dirichlet} will still be computationally involved. However, the following Theorem shows that we can efficiently and accurately approximate $\Pi_{\recht{opt}}$ by the maximum likelihood estimator:

\begin{theorem} \label{thm:mle}
For $\vec{y} \in \mathcal{B}^n$, let $\Pi_{\recht{MLE}}(\vec{y})$ be the maximum likelihood estimator of $P$ given $\vec{y}$. Then $\Pi_{\recht{MLE}} \rightarrow \Pi_{\recht{opt}}$  in probability as $n \rightarrow \infty$. 
Furthermore, one finds $\Pi_{\recht{MLE}}$ by solving an $(a-1)$-dimensional convex optimisation problem whose complexity does not depend on $n$.
\end{theorem}

\begin{proof}
The convergence in probability is proven in \cite{strasser1975asymptotic}.  Furthermore, since $\mathbb{P}(\vec{Y}=\vec{y}|P=p) = \prod_{\beta} (Q \cdot p)_\beta^{s_\beta}$, we find $\Pi_{\recht{MLE}}(\vec{y})$ by solving the optimisation problem 
\begin{equation}
\recht{min}_{p \in \mathcal{P}_{\mathcal{A}}}  \left\{ -\sum_{\beta \in \mathcal{B}} s_\beta \log((Q \cdot p)_\beta)\right\}. \label{eq:minimise}
\end{equation}
The only effect of $n$ is in the scaling of the objective function, but this does not influence the difficulty of minimisation.
\end{proof}

Since the objective function in (\ref{eq:minimise}) is smooth and convex in $p$, this can be solved quickly numerically. Using the MLE rather than the posterior also has the advantage that it is independent of the choice of $\mu$, making it a good choice of estimator when the prior is unknown.

Unfortunately, finding an expression for $\Phi_{\recht{MLE}}$ is more complicated. 
However, for large $n$ we have $F \approx P$, so we can use the optimisation problem for $P$ to get an approximate value for $F$. Note that our approach to approximating the MLE is different from that of \cite{wang2019locally}, as there the MLE is approximated based on a noninjective transformation of the obfuscated tallies $S$, rather than on $S$ itself. This transformation is protocol-specific, and hence this method does not easily extend to general protocols. The advantage of (\ref{eq:minimise}) is that it can be used to improve the estimation of any protocol.

In Section \ref{sec:exp} we numerically evaluate how well $\Pi_{\recht{MLE}}$ and works as an estimator for both $P$ and $F$.

\section{Lower bounds on MSE} \label{sec:bounds}

In this section, we prove our new lower bounds on $\recht{MSE}^{\recht{distr}}(\mathcal{Q},\Pi)$ and $\recht{MSE}^{\recht{freq}}(\mathcal{Q},\Phi)$ in terms of the LDP parameter $\varepsilon$. Our approach consists of multiple steps, but the general outline is as follows:
\begin{enumerate}
\item In Theorem \ref{thm:bound1} it is proven that one can state $\recht{MSE}^{\recht{distr}}(\mathcal{Q},\Pi)$ and $\recht{MSE}^{\recht{freq}}(\mathcal{Q},\Phi)$ as variances of $P$ and $F$ given $\vec{Y}$. We can bound these variances in terms of the entropies $\recht{h}(P|\vec{Y})$ and $\recht{H}(T|\vec{Y})$ (Theorem \ref{thm:bound2}).
\item Using results about the asymptotic information-theoretic behaviour of LDP in \cite{lopuhaa2019information}, we express the limit (as $n \rightarrow \infty)$ of these entropies in terms of properties of the matrix $Q \in \mathbb{R}^{b \times a}$ (Theorem \ref{thm:hlim}).
\item Finally, we bound these linear-algebraic constructs in terms of $\varepsilon$, giving the desired result (Theorem \ref{cor:eps}).
\end{enumerate}

We also show that the MLE bound $\frac{a}{n\varepsilon^2}$ is the optimal bound of this form, in the sense that it can be attained for $a = 2$ for both distribution and frequency estimation (Corollary \ref{cor:opt}). Note that from a theoretical perspective, the intermediate steps can be of interest as well, as they offer lower bounds on the MSE from different viewpoints: information theory, linear algebra, and the LDP parameter $\varepsilon$.

\subsection{Reduction to information theory}

The disadvantage of the optimal estimators of the previous section is that their MSEs can be hard to quantify. 
In this section, we give lower bounds for these, based on the information-theoretical quantities $\recht{h}(P|\vec{Y})$ and $\recht{h}(T|\vec{Y})$. Intuitively, the lower the uncertainty about the value of $P$ or $T$, the lower the average error on the estimation should be. While Fano's inequality expresses the same sentiment, we need the following version, which more closely aligns with the MSE rather than a binary fail/success estimator.

\begin{theorem} \label{thm:bound2}
For any privacy protocol $\mathcal{Q}$ one has
\begin{align}
\recht{MSE}_{\mu}^{\recht{distr}}(\mathcal{Q},\Pi_{\recht{opt}}) &\geq \frac{a}{2\pi\textrm{\emph{e}}}\textrm{\emph{e}}^{\frac{2}{a-1}\recht{h}(P|\vec{Y})},\\
\liminf_{n \rightarrow \infty} \frac{\recht{MSE}_{\mu}^{\recht{freq}}(\mathcal{Q},\Phi_{\recht{opt}})}{\frac{a}{2\pi\textrm{\emph{e}}n^2}\textrm{\emph{e}}^{\frac{2}{a-1}\recht{H}(F|\vec{Y})}} &\geq 1.
\end{align}
\end{theorem}

\begin{proof}
We start with $\Pi_{\recht{opt}}$. By (\ref{eq:varbound}) we have
\begin{align}
\recht{MSE}^{\recht{distr}}(\mathcal{Q},\Pi_{\recht{opt}}) &= \mathbb{E}_{\vec{y}}\sum_{x=1}^a\recht{Var}(P_x|\vec{Y}=\vec{y}) \\
&= \tfrac{1}{a-1}\mathbb{E}_{\vec{y}}\sum_{\substack{x \leq a,\\x'\neq x}}\recht{Var}(P_{x'}|\vec{Y}=\vec{y}) \\
&\geq \mathbb{E}_{\vec{y}}\sum_{x=1}^a\prod_{x'\neq x}\recht{Var}(P_{x'}|\vec{Y}=\vec{y})^{\frac{1}{a-1}}. \label{eq:hbound1}
\end{align}
Here the last inequality is the arithmetic-geometric mean inequality. Now consider the $a$-th summand; we again write $P'= (P_1,\cdots,P_{a-1})$. The $\{\recht{Var}(P_{x'}|\vec{Y}=\vec{y})\}_{x'< a}$ are the diagonal coefficients of the positive definite matrix $\recht{Cov}(P'|\vec{Y}=\vec{y})$. By Hadamard's inequality we have
\begin{equation}
\prod_{x' < a}\recht{Var}(P_{x'}|\vec{Y}=\vec{y}) \geq \recht{det} \recht{Cov}(P'|\vec{Y}=\vec{y}).
\end{equation}
Furthermore, Theorem 9.6.5 of \cite{cover1999elements} shows that
\begin{equation}
\recht{det} \recht{Cov}(P'|\vec{Y}=\vec{y}) \geq (2\pi\textrm{e})^{1-a}\textrm{e}^{2\recht{h}(P'|\vec{Y}=\vec{y})}. \label{eq:hbound3}
\end{equation}
hence
\begin{equation} \label{eq:hbound2}
\prod_{x' < a}\recht{Var}(P_{x'}|\vec{Y}=\vec{y})^{\frac{1}{a-1}} \geq \frac{1}{2\pi\textrm{e}}\textrm{e}^{\frac{2}{a-1}\recht{h}(P'|\vec{Y}=\vec{y})}.
\end{equation}
The discussion in section \ref{ssec:prob}  shows us that $\recht{h}(P'|\vec{Y}=\vec{y}) = \recht{h}(P|\vec{Y} = \vec{y})$, and that this does not depend on the enumeration of $\mathcal{A}$. It follows that (\ref{eq:hbound2}) in fact holds for every summand in (\ref{eq:hbound1}), hence we derive
\begin{equation}
\recht{MSE}^{\recht{distr}}_{\mu}(\mathcal{Q},\Pi_{\recht{opt}}) \geq \frac{a}{2\pi\textrm{e}}\mathbb{E}_{\vec{y}}\textrm{e}^{\frac{2}{a-1}\recht{h}(P|\vec{Y}=\vec{y})}.
\end{equation}
Theorem \ref{thm:bound2} is now proven by the convexity of the exponential function, which tells us that
\begin{equation}
\mathbb{E}_{\vec{y}}\textrm{e}^{\frac{2}{a-1}\recht{h}(P|\vec{Y}=\vec{y})} \geq\textrm{e}^{\frac{2}{a-1} \mathbb{E}_{\vec{y}}\recht{h}(P|\vec{Y}=\vec{y})} =  \textrm{e}^{\frac{2}{a-1}\recht{h}(P|\vec{Y})}.
\end{equation}
As for $\Phi_{\recht{opt}}$, we let $T = nF$ be as in (\ref{eq:deft}). Then $\recht{H}(F|\vec{Y}=\vec{y}) = \recht{H}(T|\vec{Y} = \vec{y})$ and $\recht{Var}(F_x|\vec{Y} = \vec{y}) = \tfrac{1}{n^2}\recht{Var}(T_x|\vec{Y} = \vec{y})$, for every $x$ and $\vec{y}$. As $n$ increases the covariance of $T$ increases, and $T$ approaches the discretisation (centered around the points of $\mathbb{Z}^a$) of a continuous random variable $\tilde{T}$. It follows that $\recht{H}(T|\vec{Y} = \vec{y}) \approx \recht{h}(\tilde{T}|\vec{Y} = \vec{y})$ for any $\vec{y}$. Analogous to the discussion on $P$ above, we can show that
\begin{align}
\mathbb{E}_{\vec{y}}\sum_{x=1}^a\recht{Var}(T_x|\vec{Y}=\vec{y}) &\approx \mathbb{E}_{\vec{y}}\sum_{x=1}^a\recht{Var}(\tilde{T}_x|\vec{Y}=\vec{y}) \\
&\geq \frac{a}{2\pi\textrm{e}}\textrm{e}^{\frac{2}{a-1}\recht{h}(\tilde{T}|\vec{Y})} \\
&\approx \frac{a}{2\pi\textrm{e}}\textrm{e}^{\frac{2}{a-1}\recht{H}(T|\vec{Y})},
\end{align}
with the approximations approaching equality as $n \rightarrow \infty$. It follows that $\recht{MSE}^{\recht{freq}}(\mathcal{Q},\Phi_{\recht{opt}}) \gtrsim \frac{a}{2\pi\textrm{e}n^2}\textrm{e}^{\frac{2}{a-1}\recht{H}(F|\vec{Y})}$.
\end{proof}

Unfortunately the result for frequencies only holds as $n \rightarrow \infty$, since it relies on approximating the discrete random variable $F$ by a continuous one. For the interests of this paper, however, this is not too much of an inconvenience, since the results from \cite{lopuhaa2019information} that we wish to apply require $n \rightarrow \infty$ in the first place.

\subsection{Accuracy bounds from linear algebra} \label{ssec:hlim}

In the previous section, we gave a lower bound for the MSE in terms of the information-theoretic quantities $\recht{h}(P|\vec{Y})$ and $\recht{H}(F|\vec{Y})$. 
We can obtain new lower bounds for the limit case by studying the behaviour of $\recht{h}(P|\vec{Y})$ and $\recht{H}(F|\vec{Y})$ as $n \rightarrow \infty$; this expands on work in \cite{lopuhaa2019information}. The resulting lower bound is weaker not in the sense that the bound is lower, but rather that it only applies to the limit case $n \rightarrow \infty$, rather than all $n$. However, the advantage is that this limit case can be formulated purely in terms of linear algebra, and as it does not depend on $n$, it is computationally more feasible for large amounts of users, which is typical in the LDP setting. 

Before we can state the result, we first need a bit more notation. We fix an identification $\mathcal{B} = \{1,\cdots,b\}$. For $x \in \mathcal{A}$, let $w_x$ be the column vector $(Q_{1|x},\cdots,Q_{b-1|x})^{\recht{T}} \in \mathbb{R}^{b-1}$. For $x \in \mathcal{A}$ and $p \in \mathcal{P}_{\mathcal{A}}$, the latter regarded as an $a$-dimensional column vector, we define matrices $D_p$, $E_x$ and $G_p$ by
\begin{align}
D_p &:= Q^{\recht{T}}\recht{diag}(Q \cdot p)^{-1}Q \in \mathbb{R}^{a\times a},\\
E_x &:= \recht{diag}(w_x) - w_xw_x^{\recht{T}} \in \mathbb{R}^{(b-1) \times (b-1)}, \label{eq:ex}\\
G_p &:= \sum_{x=1}^a p_x E_x \in \mathbb{R}^{(b-1) \times (b-1)}.
\end{align}
Furthermore, we define constants $\gamma_{\mu}(\mathcal{Q})$, $\delta_{\mu}(\mathcal{Q})$ by
\begin{align}
\gamma_{\mu}(\mathcal{Q}) &= \tfrac{a-1}{2}\log(2\pi\textrm{e})-\tfrac{1}{2}\mathbb{E}_P \log \det D_P,\\
\delta_{\mu}(\mathcal{Q}) &= \gamma_{\mu}(\mathcal{Q}) + \tfrac{1}{2}\mathbb{E}_P \log \frac{\det G_P}{\prod_{y=1}^b (Q \cdot P)_{y}}.
\end{align}

While the matrix $G_p$ depends on the choice of enumeration of $\mathcal{B}$, the resulting constant $\delta_{\mu}(\mathcal{Q})$ does not. The introduction of these constants allows us to formulate the following Theorem:

\begin{theorem} \label{thm:hlim}
It holds that
\begin{align}
\lim_{n \rightarrow \infty} \recht{h}(P|\vec{Y}) + \tfrac{a-1}{2}\log n &=  \gamma_\mu(\mathcal{Q}), \label{eq:hlim1}\\
\lim_{n \rightarrow \infty} \recht{H}(F|\vec{Y}) - \tfrac{a-1}{2}\log n &= \delta_{\mu}(\mathcal{Q}). \label{eq:hlim2}
\end{align}
\end{theorem}

\begin{proof}
Framed in the language of this section, and applying our assumption that $Q$ has rank~$a$, \cite[Thm.~6.7.1]{lopuhaa2019information} states that
\begin{equation}
\lim_{n \rightarrow \infty} \recht{I}(\vec{Y};P) - \tfrac{a-1}{2} \log n = -\gamma_{\mu}(\mathcal{Q}) + \recht{h}(P). \label{eq:hlimproof1}
\end{equation}

Since $\recht{h}(P|\vec{Y}) = \recht{h}(P)-\recht{I}(\vec{Y};P)$, this proves (\ref{eq:hlim1}). Similarly, we have $\recht{H}(F|\vec{Y}) = \recht{H}(T|\vec{Y}) = \recht{H}(T)-\recht{I}(\vec{Y};T) = \recht{H}(T) - \recht{I}(S;T)$, where $S$ and $T$ are as in (\ref{eq:deft}) and (\ref{eq:defs}). The last equation holds because given $T$, $S$ is a sufficient statistic for $\vec{Y}$. We start by describing the limit behaviour of $\recht{I}(S;T) = \recht{H}(S|P)+\recht{I}(S;P)-\recht{H}(S|T)$. By \cite[Lem.~C.3]{lopuhaa2019information}, we have
\begin{equation}
\lim_{n \rightarrow \infty} \recht{H}(S|P) - \tfrac{b-1}{2} \log n = \tfrac{b-1}{2}\log (2\pi\textrm{e}) + \frac{1}{2} \sum_{y=1}^b \mathbb{E}_P [(Q \cdot P)_y]. \label{eq:hlimproof2}
\end{equation}
Furthermore, $\recht{I}(\vec{Y};P) = \recht{I}(S;P)$, so by (\ref{eq:hlimproof1}) we get
\begin{equation}
\lim_{n \rightarrow \infty} \recht{I}(S;P) - \tfrac{a-1}{2} \log n = -\gamma_{\mu}(\mathcal{Q}) + \recht{h}(P). \label{eq:hlimproof3}
\end{equation}
It remains to study $\recht{H}(S|T)$. Let $S_{y|x}$ be as in (\ref{eq:defs2}), and let $S_{\bullet|x} = (S_{1|x},\cdots,S_{b|x}) \in \mathbb{Z}^{b}_{\geq 0}$. Then $S_{\bullet|x}$ follows a multinomial distribution with $T_x$ samples and probability vector $(Q_{1|x},\cdots,Q_{b|x})$. Let $S'_x := (S_{1|x},\cdots,S_{b-1|x}) \in \mathbb{Z}^{b-1}$. By the multivariate de Moivre--Laplace theorem \cite{veeh1986multivariate}, we know that as $n$ goes to infinity, $S'_x$ can be approximated by the discretisation of a multivariate normal distribution with mean $T_xw_x$ and covariance matrix $T_xE_x$. Applying \cite[Lem.~1.1]{ding2007eigenvalues} to the matrix $\recht{diag}(w_x)$ and the vectors $w_x$ and $-w_x$, we find
\begin{equation}
\recht{det} E_x = \prod_{y=1}^b Q_{y|x}. \label{eq:edet}
\end{equation}
Since we assumed in Section \ref{ssec:ass} that each $Q_{y|x}$ is strictly positive, this means that $E_x$ is nonsingular, and hence the associated multivariate normal distribution is nonsingular. Let $S' = (S_1,\cdots,S_{b-1})$; then $S' = \sum_{x=1}^a S'_x$, so for large $n$, given $T$, the random variable $S'$ can be approximated by a multivariate normal variable $N'$ with mean $\sum_{x=1}^a T_xw_x$ and covariance matrix $\sum_{x=1}^aT_xE_x$. Using the known formula for the differential entropy of a multivariate normal distribution \cite{cover1999elements}, it follows that for large $n$
\begin{align}
\recht{H}(S|T) &= \recht{H}(S'|T) \\
&\approx \recht{h}(N'|T) \\
&\approx \mathbb{E}_T\left[\tfrac{b-1}{2}\log(2\pi\textrm{e}) + \tfrac{1}{2} \log\det\left(\sum_{x=1}^aT_xE_x\right)\right].
\end{align}
Here $\approx$ means `the difference goes to $0$ as $n \rightarrow \infty$'. Since $n^{-1}T \approx P$ for large $n$, we find
\begin{equation}
\recht{H}(S|T) \approx \tfrac{b-1}{2}\log(2\pi\textrm{e}n) + \mathbb{E}_P\left[\tfrac{1}{2} \log\det G_P\right]. \label{eq:hlimproof4}
\end{equation}
Combining (\ref{eq:hlimproof2}), (\ref{eq:hlimproof3}) and (\ref{eq:hlimproof4}) now finishes the proof of (\ref{eq:hlim2}).
\end{proof}

Combining Theorems \ref{thm:bound2} and \ref{thm:hlim} allows us bound the MSE in terms of $\gamma_{\mu}(\mathcal{Q})$ and $\delta_{\mu}(\mathcal{Q})$.

\begin{corollary} \label{cor:linalg}
One has
\begin{align}
\liminf_{n \rightarrow \infty} n\recht{MSE}^{\recht{distr}}_{\mu}(\mathcal{Q},\Pi_{\recht{opt}}) &\geq \frac{a}{2\pi\textrm{\emph{e}}}\textrm{\emph{e}}^{\frac{2}{a-1}\gamma_{\mu}(\mathcal{Q})}, \label{eq:linalg1}\\
\liminf_{n \rightarrow \infty} n\recht{MSE}^{\recht{freq}}_{\mu}(\mathcal{Q},\Phi_{\recht{opt}}) &\geq \frac{a}{2\pi\textrm{\emph{e}}}\textrm{\emph{e}}^{\frac{2}{a-1}\delta_{\mu}(\mathcal{Q})}. \label{eq:linalg2}
\end{align}
\end{corollary}

\begin{proof}
One has $\lim_{n \rightarrow \infty} n\frac{a}{2\pi\textrm{e}}\textrm{e}^{\frac{2}{a-1}\recht{h}(P|\vec{Y})} =  \frac{a}{2\pi\textrm{\emph{e}}}\textrm{\emph{e}}^{\frac{2}{a-1}\gamma_{\mu}(\mathcal{Q})}$ by Theorem \ref{thm:hlim}. On the other hand, by Theorem \ref{thm:bound2}, one has 
\begin{equation}
\liminf_{n \rightarrow \infty}  n\recht{MSE}^{\recht{distr}}_{\mu}(\mathcal{Q},\Pi_{\recht{opt}}) \geq \liminf_{n \rightarrow \infty} n\frac{a}{2\pi\textrm{e}}\textrm{e}^{\frac{2}{a-1}\recht{h}(P|\vec{Y})}.
\end{equation}
 Combining these two statements proves (\ref{eq:linalg1}). Equation (\ref{eq:linalg2}) can be proven analogously.
\end{proof}

We should interpret this Corollary as stating that in the best case we have $\recht{MSE}^{\recht{distr}}_{\mu}(\mathcal{Q},\Pi_{\recht{opt}}), \recht{MSE}^{\recht{freq}}_{\mu}(\mathcal{Q},\Phi_{\recht{opt}}) = \Omega(n^{-1})$ for fixed $\mathcal{Q}$ and $a$, and we can bound the constants involved.

\subsection{Accuracy bounds from $\varepsilon$-LDP}

While the constants $\gamma_{\mu}(\mathcal{Q})$ and $\delta_{\mu}(\mathcal{Q})$ do not depend on $n$, they still involve integration over $\mathcal{P}_{\mathcal{A}}$, and as such can be computationally difficult for large $a$. However, it is possible to give lower bounds for these constants that are independent of $\mu$, and whose $\mathcal{Q}$-dependence only appears in the privacy parameter $\varepsilon$.

\begin{theorem} \label{thm:eps}
Suppose $\mathcal{Q}$ satisfies $\varepsilon$-LDP. Then
\begin{align}
\gamma_{\mu}(\mathcal{Q}) &\geq (a-1)\log\frac{\sqrt{2\pi\textrm{\emph{e}}}}{\textrm{\emph{e}}^{\varepsilon}-1}, \label{eq:epsthm1}\\
\delta_{\mu}(\mathcal{Q}) &\geq (a-1)\log\frac{\sqrt{2\pi\textrm{\emph{e}}}}{\textrm{\emph{e}}^{\varepsilon}-1}-\frac{b\varepsilon}{2}. \label{eq:epsthm2}
\end{align}
\end{theorem}

\begin{proof}
In the terminology of the present paper, $\recht{U}^{\recht{as}}_{\mu}(\mathcal{Q})$ of \cite[Def.~6.5]{lopuhaa2019information} equals $-\frac{1}{a-1}\gamma_{\mu}(\mathcal{Q})$, and $\recht{S}^{\recht{wc}}(\mathcal{Q})$ of \cite[Def.~5.1]{lopuhaa2019information} equals $\textrm{e}^{-\varepsilon'}$, where $\varepsilon' \leq \varepsilon$ is minimal such that $\mathcal{Q}$ satisfies $\varepsilon'$-LDP. In this terminology, \cite[Prop.~8.1]{lopuhaa2019information} tells us that 
\begin{equation}
\gamma_{\mu}(\mathcal{Q}) \geq (a-1)\log\frac{\sqrt{2\pi\textrm{e}}}{\textrm{e}^{\varepsilon'}-1} \geq (a-1)\log\frac{\sqrt{2\pi\textrm{e}}}{\textrm{e}^{\varepsilon}-1},
\end{equation}
proving (\ref{eq:epsthm1}). As for (\ref{eq:epsthm2}), by the definition of $\varepsilon$-LDP, for every $x \in \mathcal{A}$ and every $p \in \mathcal{P}_{\mathcal{A}}$ we have $\prod_{y=1}^b Q_{y|x} \geq \textrm{e}^{-b\varepsilon}\prod_{y=1}^b (Q\cdot p)_y$. Let $E_x$ be as in (\ref{eq:ex}). By (\ref{eq:edet}) we get 
\begin{equation}
\recht{det} E_x = \prod_{y=1}^b Q_{y|x} \geq \textrm{e}^{-b\varepsilon}\prod_{y=1}^b (Q\cdot p)_y
\end{equation}
for every $x$ and $p$. Since $\log \det$ is concave on the space of positive symmetric matrices, we find for every $p$ that 
\begin{equation}
\log \det G_p \geq \sum_x p_x \log \det E_x \geq -b\varepsilon+\sum_{y=1}^b \log(Q\cdot p)_y.
\end{equation}
Combined with the definition of $\delta_{\mu}(\mathcal{Q})$ this now directly proves (\ref{eq:epsthm2}).
\end{proof}

As a corollary of this Theorem, we find a lower bound for the MSE in terms of $\varepsilon$. We write $f(\varepsilon) \succeq g(\varepsilon)$ if $\liminf_{\varepsilon \rightarrow 0} \frac{f(\varepsilon)}{g(\varepsilon)} \geq 1$. This Theorem follows directly from substituting the bounds for $\gamma_{\mu}(\mathcal{Q})$ and $\delta_{\mu}(\mathcal{Q})$ from Theorem \ref{thm:eps} into Corollary \ref{cor:linalg}.

\begin{theorem} \label{cor:eps}
Let $\Pi_{\recht{opt}}$ and $\Phi_{\recht{opt}}$ be as in Theorem \ref{thm:bound2}. Then for every $\mu$ one has 
\begin{align}
\liminf_{n \rightarrow \infty} n\recht{MSE}_{\mu}^{\recht{distr}}(\mathcal{Q},{\Pi}_{\recht{opt}}) &\geq \frac{a}{(\textrm{\emph{e}}^{\varepsilon}-1)^2}, \label{eq:epsbound1} \\
\liminf_{n \rightarrow \infty} \recht{MSE}_{\mu}^{\recht{freq}}(\mathcal{Q},{\Phi}_{\recht{opt}}) &\geq \frac{a\textrm{\emph{e}}^{-\frac{b}{2(a-1)}\varepsilon}}{(\textrm{\emph{e}}^{\varepsilon}-1)^2}. \label{eq:epsbound2}
\end{align}
For $n \in \mathbb{Z}_{>0}$ and $\varepsilon \in \mathbb{R}_{>0}$, let $(\mathcal{Q}_{n,\varepsilon},\Pi_{\recht{opt}})$ be the pair of an $\varepsilon$-LDP privacy protocol and an estimator for $P$ minimising (\ref{eq:msedef1}), and let $(\mathcal{Q}'_{n,\varepsilon},\Phi_{\recht{opt}}))$ be the pair of an $\varepsilon$-LDP privacy protocol and an estimator for $T$ minimising (\ref{eq:msedef2}). Then as $\varepsilon \rightarrow 0$, we have
\begin{align}
\liminf_{n \rightarrow \infty} n\recht{MSE}_{\mu}^{\recht{distr}}(\mathcal{Q}_{n,\varepsilon},\Pi_{\recht{opt}}) &\succeq \tfrac{a}{\varepsilon^2},\\
\liminf_{n \rightarrow \infty} n\recht{MSE}_{\mu}^{\recht{freq}}(\mathcal{Q}'_{n,\varepsilon},\Phi_{\recht{opt}}) &\succeq\tfrac{a}{\varepsilon^2}.
\end{align}
\end{theorem}

Note that these results are strictly better than what is known in the literature for the case $n \rightarrow \infty$, to the best of our knowledge: (\ref{eq:epsbound1}) improves the result in the proof of Proposition 6 of \cite{duchi2014local} in three ways: First, our result does not just give a bound for the WMSE, but also for the MSE. Also, we improve the lower bound by a factor $64$. Furthermore, the result of \cite{duchi2014local} only holds for $\varepsilon \leq 1$. The downside of our result, however, is that it only holds for the limit case $n \rightarrow \infty$.

As for the results for $\Phi_n$, it follows from results in \cite{blasiok2019towards} that the optimal $(Q,\Phi_{n,\varepsilon})$ satisfies $\recht{WMSE}^{\recht{freq}}(\mathcal{Q}',\Phi_{n,\varepsilon}) = \Omega(\tfrac{a}{n\varepsilon^2})$ for small $\varepsilon$, but the authors do not make a statement about the constants involved. Equation (\ref{eq:epsbound2}) improves upon this by giving a quantitative lower bound for the MSE, which itself is a lower bound for the WMSE. Also, the OUE and OLH protocols from \cite{wang2017locally} performs as $\approx \tfrac{4a}{n\varepsilon^2}$ for large $n$ and small $\varepsilon$, while our lower bound is of the form $\tfrac{a}{n\varepsilon^2}$. Therefore, our bound is quite near to what is possible in practice.

The bound in (\ref{eq:epsbound2}) seems to imply that when looking for optimal $\mathcal{Q}$, we should take $b$ as large as possible. This seemingly contradicts results in \cite{kairouz2014extremal}, where it is found that taking $b = a$ is always sufficient. There are two possible explanations for this. First, it is possible that our bounds are not sharp enough to accurately detect the dependence on $b$; this is especially probable since (\ref{eq:epsbound2}) is only the latest in a chain of inequalities. However, the discrepancy between our results and those of \cite{kairouz2014extremal} can also be caused by the fact that different utility metrics were being used. Whereas we focus on asymptotically many users, the utility metric in \cite{kairouz2014extremal} looks at the KL-divergence between the probability distributions induced by the private datum of one user. It is a possibility that the optimal protocols for these different metrics do not coincide. Overall, the $b$-dependence of the estimation error is difficult to assess, as there exist protocols achieving the best known MSE $\frac{4a}{n\varepsilon^2}$ for both $b = 2a$ \cite{acharya2019hadamard} and $b = 2^a$ \cite{duchi2014local}.

\subsection{Tightness for $a=2$}

In this section, we show that if $a=2$ and as $\varepsilon \rightarrow 0$, the bounds in Theorem \ref{cor:eps} are tight. For this, we recall the Randomised Response protocol \cite{warner1965randomized} for $a=2$, which, for an $\varepsilon > 0$, is the LDP protocol $\recht{RR}_{\varepsilon}\colon\{1,2\} \rightarrow \{1,2\}$ given by the matrix
\begin{equation}
    \left(\begin{array}{cc}
    \frac{\textrm{e}^{\varepsilon}}{\textrm{e}^{\varepsilon}+1} & \frac{1}{\textrm{e}^{\varepsilon}+1} \\
    \frac{1}{\textrm{e}^{\varepsilon}+1} & \frac{\textrm{e}^{\varepsilon}}{\textrm{e}^{\varepsilon}+1}\end{array}\right).
\end{equation}
Note that $\recht{RR}_{\varepsilon}$ satisfies $\varepsilon$-LDP. Let $s_1,s_2$ be as in (\ref{eq:defs}). As estimators for $P$ and $T$, respectively, we use the maps $\Pi_{\recht{RR}},\Phi_{\recht{RR}}\colon \{1,2\}^n \rightarrow \mathbb{R}^2$ given by
\begin{align}
\Pi_{\recht{RR}}(\vec{y}),\Phi_{\recht{RR}}(\vec{y}) &= \left(\frac{(\textrm{e}^{\varepsilon}+1)s_1-1}{n(\textrm{e}^{\varepsilon}-1)},\frac{(\textrm{e}^{\varepsilon}+1)s_2-1}{n(\textrm{e}^{\varepsilon}-1)}\right). \label{eq:fndef}
\end{align}
These are unbiased estimators for $P$ and $F$ \cite{warner1965randomized}.

\begin{proposition} \label{prop:rr}
For any prior $\mu$ of $P$ one has
\begin{align}
\recht{MSE}_{\mu}^{\recht{distr}}(\recht{RR}_{\varepsilon},\Pi_{\recht{RR}}) &= \tfrac{2}{n}\left(\tfrac{\textrm{\emph{e}}^{\varepsilon}}{(\textrm{\emph{e}}^{\varepsilon}-1)^2}+\mathbb{E}_{P}[P_1P_2]\right), \label{eq:fnprop}\\
\recht{MSE}_{\mu}^{\recht{freq}}(\recht{RR}_{\varepsilon},\Phi_{\recht{RR}}) &= \frac{2\textrm{\emph{e}}^{\varepsilon}}{n(\textrm{\emph{e}}^{\varepsilon}-1)^2}. \label{eq:gnprop}
\end{align}
\end{proposition}

\begin{proof}[Proof of Proposition \ref{prop:rr}]
Since $\Pi_{\recht{RR}}$ is an unbiased estimator of $P$ and $\Pi_{\recht{RR}}(\vec{Y})_1 + \Pi_{\recht{RR}}(\vec{Y})_2 = 1$, we have
\begin{align}
\recht{MSE}_{\mu}^{\recht{distr}}(\recht{RR}_{\varepsilon},\Pi_{\recht{RR}})
&= \mathbb{E}_p\recht{Var}(\Pi_{\recht{RR}}(\vec{Y})_1|P=p) \\ 
& \ \ \ +\mathbb{E}_p\recht{Var}(\Pi_{\recht{RR}}(\vec{Y})_2|P=p) \nonumber\\
&= 2\mathbb{E}_p\recht{Var}(\Pi_{\recht{RR}}(\vec{Y})_1|P=p).
\end{align}
For a given $P = p$, we know that $S_1$ is binomially distributed with $n$ samples and probability $\frac{1+(\textrm{e}^{\varepsilon}-1)p_1}{\textrm{e}^{\varepsilon}+1}$. Substituting this into (\ref{eq:fndef}) yields
\begin{align}
&\recht{Var}(\Pi_{\recht{RR}}(\vec{Y})_1|P=p) \nonumber \\
&= \frac{(\textrm{e}^{\varepsilon}+1)^2}{n^2(\textrm{e}^{\varepsilon}-1)^2}\recht{Var}(S_1|P=p)\\
&= \frac{(\textrm{e}^{\varepsilon}+1)^2}{n^2(\textrm{e}^{\varepsilon}-1)^2}\cdot n \cdot \frac{1+(\textrm{e}^{\varepsilon}-1)p_1}{\textrm{e}^{\varepsilon}+1} \cdot \frac{\textrm{e}^{\varepsilon}-(\textrm{e}^{\varepsilon}-1)p_1}{\textrm{e}^{\varepsilon}+1}\\
&= \frac{\textrm{e}^{\varepsilon}+(\textrm{e}^{\varepsilon}-1)^2p_1(1-p_1)}{n(\textrm{e}^{\varepsilon}-1)^2}.
\end{align}
Equation (\ref{eq:fnprop}) follows directly from this. Since $\Phi_{\recht{RR}}$ is an unbiased estimator of $F = n^{-1}T$, we analogously find that we only need to determine $\recht{Var}(\Phi_{\recht{RR}}(\vec{Y})_1|T=t)$. Let $S_{1|1},S_{1|2}$ be as in (\ref{eq:defs2}). Then $S_1 = S_{1|1}+S_{1|2}$, and $S_{1|1}$ is binomially distributed with $t$ samples and probability $\frac{\textrm{e}^{\varepsilon}}{\textrm{e}^{\varepsilon}+1}$, while $S_{1|2}$ is binomially distributed with $n-t$ samples and probability $\frac{1}{\textrm{e}^{\varepsilon}+1}$. Since $S_{1|1}$ and $S_{1|2}$ are independent given $T$, it follows from (\ref{eq:fndef}) that
\begin{align}
&\recht{Var}(\Phi_{\recht{RR}}(\vec{Y})_1|T=t) \nonumber \\
&= \frac{(\textrm{e}^{\varepsilon}+1)^2}{(\textrm{e}^{\varepsilon}-1)^2}\left(\recht{Var}(S_{1|1}|T=t)+\recht{Var}(S_{1|2}|T=t) \right) \\
&= \frac{(\textrm{e}^{\varepsilon}+1)^2}{(\textrm{e}^{\varepsilon}-1)^2}\left(\frac{t\textrm{e}^{\varepsilon}}{(\textrm{e}^{\varepsilon}+1)^2}+\frac{(n-t)\textrm{e}^{\varepsilon}}{(\textrm{e}^{\varepsilon}+1)^2}\right) \\
&= \frac{n\textrm{e}^{\varepsilon}}{(\textrm{e}^{\varepsilon}-1)^2}.
\end{align}
It follows that
\begin{equation}
\recht{MSE}_{\mu}^{\recht{tally}}(\recht{RR}_{\varepsilon},\Phi_{\recht{RR}}) = 2\mathbb{E}_t \recht{Var}(\Phi_{\recht{RR}}(\vec{Y})_1|T=t) = \frac{2n\textrm{e}^{\varepsilon}}{(\textrm{e}^{\varepsilon}-1)^2}. \qedhere
\end{equation}
\end{proof}

This Proposition yields the following Corollary. Below, we use $f(\varepsilon) \sim g(\varepsilon)$ to denote $\lim_{\varepsilon \rightarrow 0} \frac{f(\varepsilon)}{g(\varepsilon)} = 1$.

\begin{corollary} \label{cor:opt}
Let $a = 2$, and let $\mu$ be given. Let $(\mathcal{Q}_{n,\varepsilon},\Pi_{\recht{opt}})$ and $(\mathcal{Q}'_{n,\varepsilon},\Phi_{\recht{opt}})$ be as in Theorem \ref{cor:eps}. Then as $\varepsilon \rightarrow 0$ we have
\begin{align}
\lim_{n \rightarrow \infty} n\recht{MSE}_{\mu}^{\recht{distr}}(\mathcal{Q}_{n,\varepsilon},\Pi_{n,\varepsilon}) &\sim \tfrac{2}{\varepsilon^2},\\
\lim_{n \rightarrow \infty} n\recht{MSE}_{\mu}^{\recht{freq}}(\mathcal{Q}'_{n,\varepsilon},\Phi_{n,\varepsilon}) &\sim \tfrac{2}{\varepsilon^2}.
\end{align}
\end{corollary}

\begin{proof}
A lower bound (of the behaviour in $\varepsilon$ as $\varepsilon \rightarrow 0$) is provided by Theorem \ref{cor:eps}, while an upper bound is provided by Proposition \ref{prop:rr}.
\end{proof}

In other words, the bounds in Theorem \ref{cor:eps} become tight when $\varepsilon \rightarrow 0$, provided that $a=2$. This shows that our bounds in Table \ref{tab:results} give the best possible coefficient for $\frac{a}{n\varepsilon^2}$ that holds for all $a$.

\section{MLE experiments} \label{sec:exp}

{\bf Synthentic dataset.} We perform synthetic experiments to see how the estimator $\Pi_{\recht{MLE}}$ from Section \ref{ss:mle} performs in practice. We apply the MLE to two well-established privacy protocols, Randomised Response (RR) \cite{warner1965randomized} and Unary Encoding (UE) \cite{wang2017locally}; the latter is one of the protocols that obtains the known optimal error $\frac{4a}{n\varepsilon^2}$ for $\hat{P}$ and $\hat{F}$. Both of these protocols are parametrised by their LDP parameter $\varepsilon$. 
We take $\varepsilon \in [0.2,2]$. 
Since $b = 2^a$ for UE, the number of summands in (\ref{eq:minimise}) grows too large to handle for large $a$. Therefore we take $a=10$ for UE.
For RR $a = b$, so we do not have this problem, and we take the more general setting $a = 1024$. 
We consider both large number of samples ($n \gg a$) and moderate number of samples ($n \approx 10 a$) scenarios. 
For $\mu$, we take the Jeffreys prior on $\mathcal{P}_{\mathcal{A}}$, i.e. the symmetric Dirichlet distribution with parameter $-\tfrac{1}{2}$.
We take this prior because it is noninformative, as its definition does not depend on the parametrisation of $\mathcal{P}_{\mathcal{A}}$. We draw $P$ from $\mu$ 100 times and generate a dataset of $n$ users from it. We then randomise the data via the LDP protocol, and perform MLE on the outcome to produce an estimate $\hat{P}$ of $P$; we furthermore use $\hat{F} = \hat{P}$ as an estimator of $F$. Finally, we measure $||P-\hat{P}||_2^2$ and $||F-\hat{F}||_2^2$, and average this over all samples to determine the MSE. To calculate the MLE, we use projected gradient descent to solve (\ref{eq:minimise}) for UE. For RR, there is a direct method to find $\hat{P}$, see Appendix \ref{app:mle}.

We compare the MLE estimator to two other estimators: First, the baseline Frequency Oracles (FO), which are affine transformations of $S$ used to produce unbiased estimators of $P$ and $F$. Second, we look at Norm-Sub, which was found in \cite{wang2019locally} to be postprocessing method of the FO outcome that gives the best MSE, among a wide selection of considered postprocessing methods.

\begin{figure*}
    \centering
    \begin{subfigure}[b]{\textwidth}
		\includegraphics[width=\textwidth]{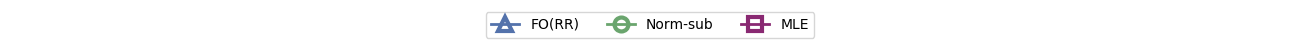}
		\vspace{-0.5cm}
		\label{RR_vary_eps_legend}
	\end{subfigure}\\
    \begin{subfigure}[b]{0.23\textwidth}
		\includegraphics[width=\textwidth]{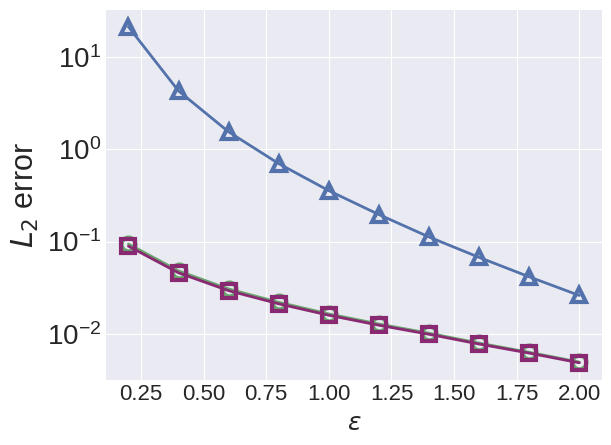}
		\vspace{-0.7cm}
		\caption{Error of $\hat{P}$ when $n=1000000$.}
		\label{RR_vary_eps_P6}
	\end{subfigure}
	\begin{subfigure}[b]{0.23\textwidth}
		\includegraphics[width=\textwidth]{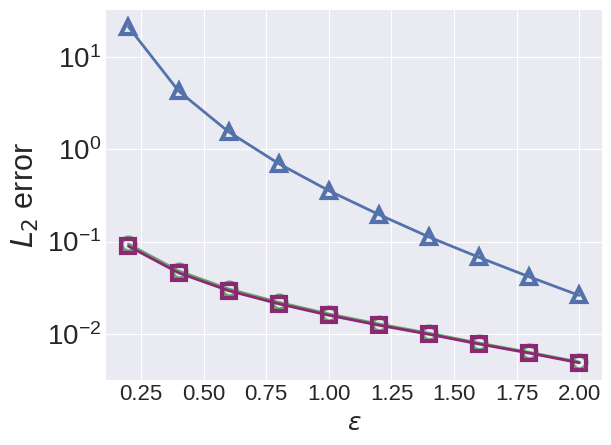}
		\vspace{-0.7cm}
		\caption{Error of $\hat{F}$ when $n=1000000$.}
		\label{RR_vary_eps_T6}
	\end{subfigure}
    \begin{subfigure}[b]{0.23\textwidth}
		\includegraphics[width=\textwidth]{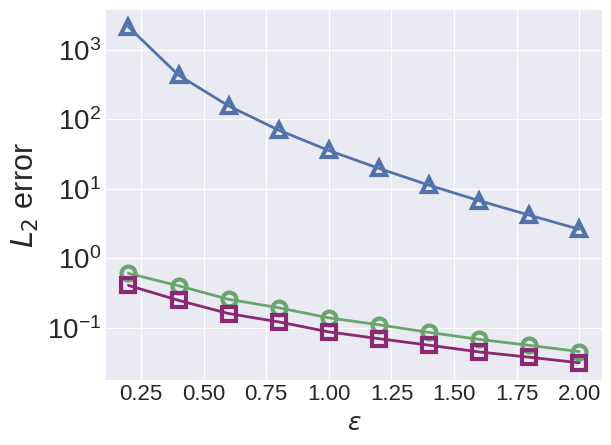}
		\vspace{-0.7cm}
		\caption{Error of $\hat{P}$ when $n=10000$.}
		\label{RR_vary_eps_P4}
	\end{subfigure}
	\begin{subfigure}[b]{0.23\textwidth}
		\includegraphics[width=\textwidth]{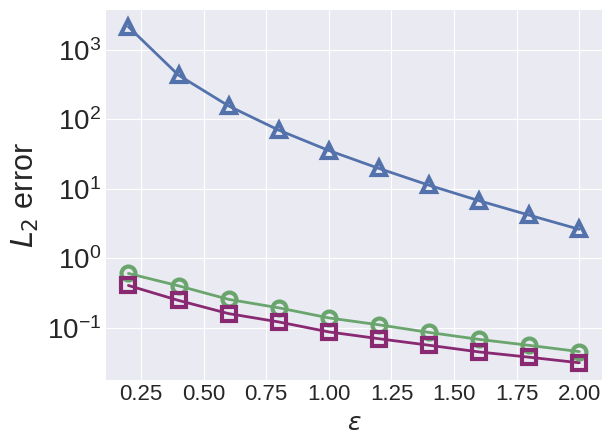}
		\vspace{-0.7cm}
		\caption{Error of $\hat{F}$ when $n=10000$.}
		\label{RR_vary_eps_T4}
	\end{subfigure}
	\caption{Experiments with RR, $a=1024$ and Jeffreys prior $\mu$.}
	\label{fig:mle_RR}
\end{figure*}

\begin{figure*}
    \centering
    \begin{subfigure}[b]{\textwidth}
		\includegraphics[width=\textwidth]{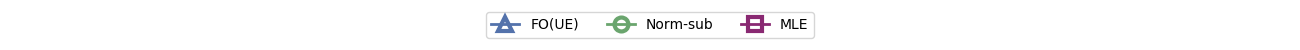}
		\vspace{-0.5cm}
	\end{subfigure}\\
    \begin{subfigure}[b]{0.23\textwidth}
		\includegraphics[width=\textwidth]{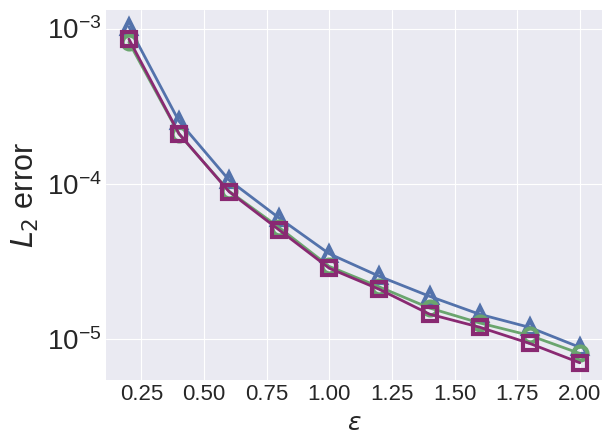}
		\vspace{-0.7cm}
		\caption{Error of $\hat{P}$ when $n=1000000$.}
		\label{ue_vary_eps_P6}
	\end{subfigure}
	\begin{subfigure}[b]{0.23\textwidth}
		\includegraphics[width=\textwidth]{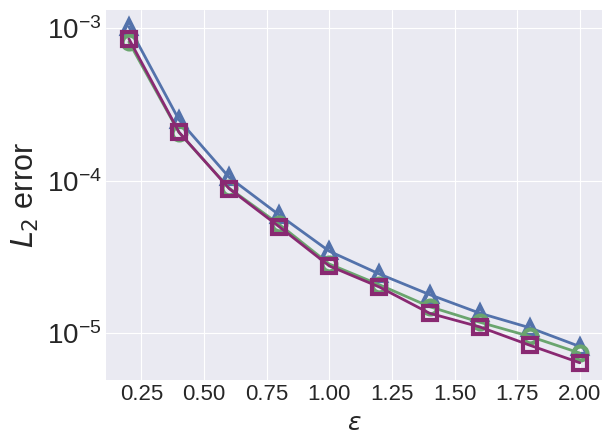}
		\vspace{-0.7cm}
		\caption{Error of $\hat{F}$ when $n=1000000$.}
		\label{ue_vary_eps_T6}
	\end{subfigure}
	\begin{subfigure}[b]{0.23\textwidth}
		\includegraphics[width=\textwidth]{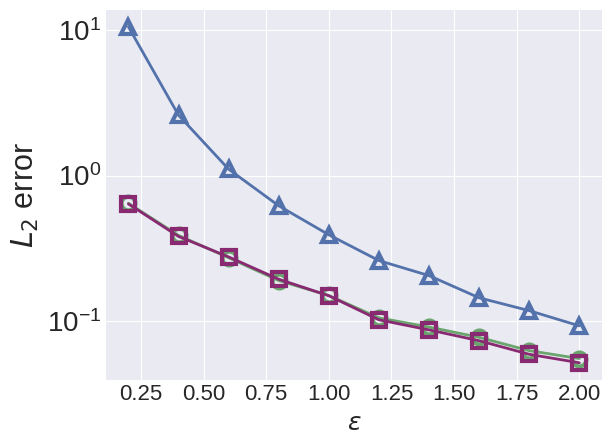}
		\vspace{-0.7cm}
		\caption{Error of $\hat{P}$ when $n=100$.}
		\label{ue_vary_eps_P4}
	\end{subfigure}
	\begin{subfigure}[b]{0.23\textwidth}
		\includegraphics[width=\textwidth]{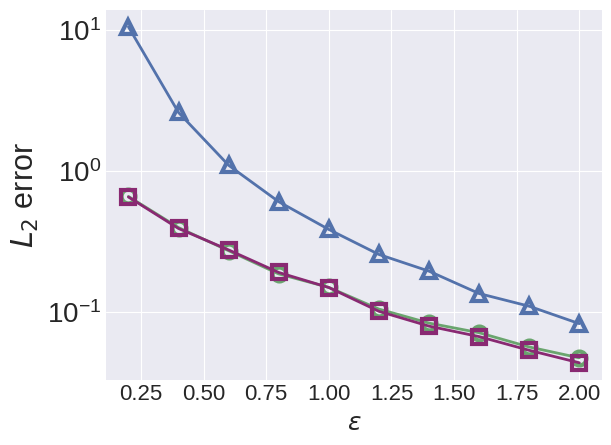}
		\vspace{-0.7cm}
		\caption{Error of $\hat{F}$ when $n=100$.}
		\label{ue_vary_eps_T4}
	\end{subfigure}
	\caption{Experiments with UE, $a=10$ and Jeffreys prior $\mu$.}
	\label{fig:mle_ue}
\end{figure*}

\begin{figure*}
    \centering
    \begin{subfigure}[b]{\textwidth}
		\includegraphics[width=\textwidth]{figure/FO_RR_.png}
		\vspace{-0.5cm}
	\end{subfigure}\\
    \begin{subfigure}[b]{0.40\textwidth}
		\includegraphics[width=\textwidth]{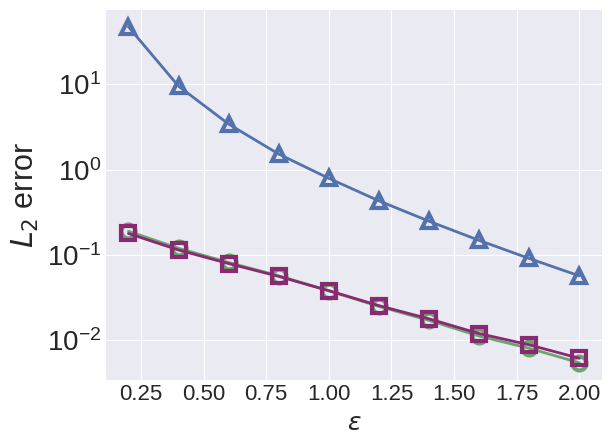}
		\vspace{-0.7cm}
		\caption{Error of RR $\hat{F}$, taxi distance, $n=434195, a=1000$.}
		\label{rr_vary_eps_P_taxi_distance}
	\end{subfigure}
	\hspace{0.5cm}
	\begin{subfigure}[b]{0.40\textwidth}
		\includegraphics[width=\textwidth]{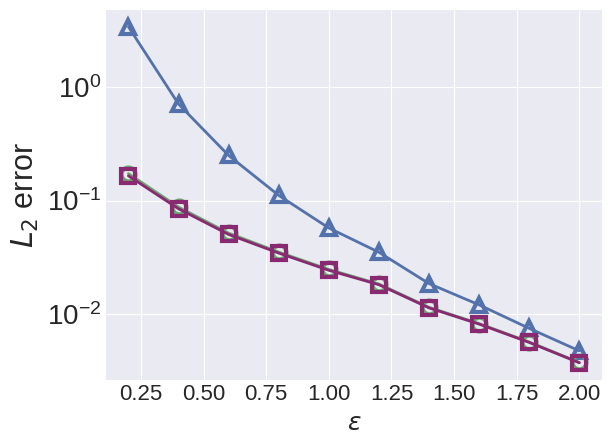}
		\vspace{-0.7cm}
		\caption{Error of RR $\hat{F}$, age in Adult, $n=32561, a=75$.}
		\label{rr_vary_eps_P_adult_country}
	\end{subfigure}
	\caption{Experiments of RR with real world datasets.}
	\label{fig:real_world_RR}
\end{figure*}

\begin{figure*}
    \centering
    \begin{subfigure}[b]{\textwidth}
		\includegraphics[width=\textwidth]{figure/FO_UE_.png}
		\vspace{-0.5cm}
	\end{subfigure}\\
	\begin{subfigure}[b]{0.40\textwidth}
		\includegraphics[width=\textwidth]{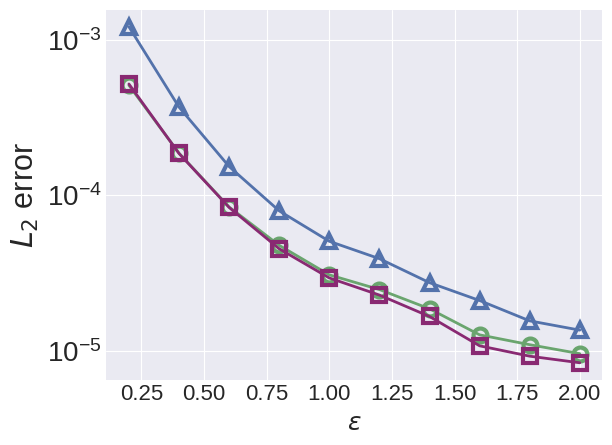}
		\vspace{-0.7cm}
		\caption{Error of UE $\hat{F}$, taxi payment type, $n=359902, a=5$.}
% 		\hspace{1cm}
		\label{ue_vary_eps_P_taxi_payment_type}
	\end{subfigure}
	\begin{subfigure}[b]{0.40\textwidth}
		\includegraphics[width=\textwidth]{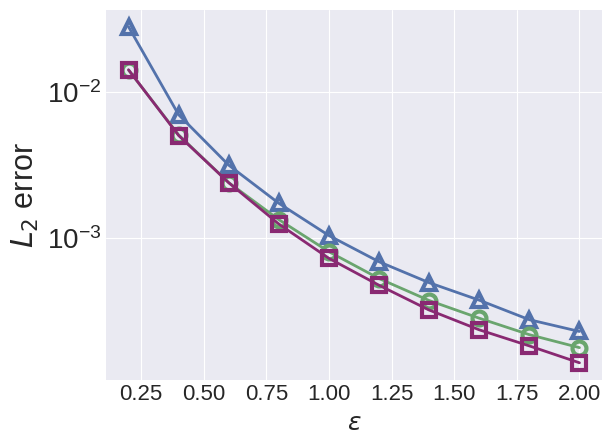}
		\vspace{-0.7cm}
		\caption{Error of RR $\hat{F}$, adult work class, $n=32561, a=8$.}
		\label{ue_vary_eps_P_adult_workclass}
	\end{subfigure}
	\caption{Experiments of UE with real world datasets.}
	\label{fig:real_world_UE}
\end{figure*}

Figures~\ref{fig:mle_RR} (RR) and \ref{fig:mle_ue} (UE) show the experimental results.
The simulation results show that both Norm-sub and MLE post-processing can elevate the accuracy of the results.
For RR, we see that Norm-Sub and MLE give the same accuracy when $n \gg a$. This is not unexpected, as the results in \cite{wang2019locally} show that for frequency estimation, Norm-Sub gives similar accuracy to MLE, and for RR the MLEs defined here and in that paper are equivalent.
However, when $n \approx 10 a$, MLE results are better than Norm-Sub, suggesting that when we use RR and $n$ is comparable to $a$, MLE can give better accuracies.
For UE, we see that MLE gives more accurate frequency estimations than Norm-sub when $\varepsilon > 1.5$ when $n \gg a$.
For a moderate number of samples cases, the results of MLE and Norm-Sub are similar.

{\bf Real world datasets.}
To show that the MLE estimator can help to improve the estimation on real world datasets, we also apply our method on two real world datasets: the taxi dataset~\cite{taxi_data} and the adult dataset~\cite{Dua:2019}.
We use discretized taxi distance ($a=1000$, valid samples $n=434195$) in the taxi dataset and age ($n=32561, a=75$) in the adult dataset for RR, taxi payment types ($a=5$, valid samples $n=359902$) and work classes ($n=32561, a=8$) in adult dataset for UE.
We choose different attributes for RR and UE because of the same reason as in synthetic experiments (the output domain size $b = 2^a$ for UE).
Since the true probability $P$ of real world datasets is unknown, we can only compare the frequency error, $\|F - \hat{F}\|^2_2$. 
All experiments are repeated $100$ times and the error means are shown in Figure~\ref{fig:real_world_RR} and \ref{fig:real_world_UE}.

The experiment results in Figure~\ref{fig:real_world_RR} (RR) and \ref{fig:real_world_UE} (UE) show that, as expected, both Norm-Sub and MLE post-processing can improve the accuracy of the results.
MLE and Norm-Sub with RR give similar accuracy on both dataset, while MLE with UE can give more accurate results than Norm-Sub with UE.

\section{Further research}

Although we have shown that our bounds are tight for $a=2$, $n \gg 0$ and $\varepsilon \ll 1$, it would be interesting to know what frequency estimation accuracy is achievable in more general settings. In particular, it is interesting to know what happens for large $a$, and for large $\varepsilon$ (i.e. in the low privacy domain). For large $\varepsilon$, the dependence on $b$ will probably also come into play, which might be able to help point us towards optimal protocols.

Another useful approach would be to give computationally feasible approximations to $\hat{\Phi}_n$ and $\hat{\Pi}_n$. The formulas of Proposition \ref{prop:dirichlet}, which are presented explicitely in Appendix \ref{app:dirichlet}, can become too complex for large $a,b,n$, and while the MLE approach reduces computational complexity, it is still too involved for large $a$. This is important, as for small $a$ the MLE is shown to give accurate frequency estimations; it would be interesting to find a way to extend this approach to larger $a$. Furthermore, both our Theorem \ref{thm:bound1}, as well as empirical results \cite{jia2019calibrate,wang2019locally}, show that prior knowledge about the distribution can have a significant impact on frequency estimation. The MLE ignores this prior knowledge, hence would be good to find a way to enhance the MLE estimation by taking prior knowledge into account.

By their asymptotic nature, our results mainly concern the case where $n \gg a$. However, LDP is also used in situations where $a \approx n$ \cite{erlingsson2014rappor}. In such situations, one is more interested in identifying the most frequent items rather than estimating the frequency of all items. Different information-theoretical techniques are needed to estimate the utility in this scenario.

\subsection*{Acknowledgments}
This project is supported by NSF grant 1640374, NWO grant 628.001.026, and NSF grant 1931443. We thank the anonymous reviewers for their helpful suggestions.

\printbibliography

\appendix

\section{Proof of Proposition \ref{prop:dirichlet}} \label{app:dirichlet}

Before we can prove Proposition \ref{prop:dirichlet}, we recall two well-known facts about Dirichlet distributions and distribution mixtures, and prove an auxiliary Proposition.

\begin{fact} \label{fact:dir}
Let $P \in \mathcal{P}_{\mathcal{A}}$ be drawn from a Dirichlet distribution with parameter vector $\gamma \in \mathbb{R}^{a}_{\geq 0}$. Then
\begin{align}
\mathbb{E}[P]&= \frac{1}{\sum_{\alpha \in \mathcal{A}}\gamma_{\alpha}}\gamma,\\
\recht{Var}(P_{\alpha}) &= \frac{\mathbb{E}[P_\alpha](1-\mathbb{E}[P_{\alpha}])}{1+\sum_{\alpha' \in \mathcal{A}}\gamma_{\alpha'}}.
\end{align}
\end{fact}

\begin{fact} \label{fact:mix}
Let $P_1,\cdots,P_r$ be continuous random variables in $\mathbb{R}$, and let $w \in \mathbb{R}^r_{\geq 0}$ be such that $\sum_{j=1}^r w_j = 1$. Let $M$ be the mixture of the $P_j$ with weight vector $w$, i.e. $M = P_j$ with probability $w_j$. Then
\begin{align}
\mathbb{E}[M] &= \sum_{j=1}^r w_j \mathbb{E}[P_j],\\
\recht{Var}(M) &= \sum_{j=1}^r w_j (\recht{Var}(P_j)+\mathbb{E}[P_j]^2-\mathbb{E}[M]^2).
\end{align}
\end{fact}

The following proposition gives explicit, if lengthy, formulas to calculate $\Pi_{\recht{opt}}$ as well as $\recht{MSE}_{\mu}^{\recht{distr}}(\mathcal{Q},\Pi_{\recht{opt}})$.

\begin{proposition} \label{prop:dir}
Let $\Pi_{\recht{opt}}$ be as in Theorem \ref{thm:bound1}. Let $\mu$ be a Dirichlet distribution with parameter vector $\gamma \in \mathbb{R}^{a}_{\geq 0}$. Let $\recht{B}$ be the multivariate beta function. For $\alpha \in \mathcal{A}$, $\vec{x} \in \mathcal{A}^n$ and $\vec{y} \in \mathcal{B}^n$, we define
\begin{align}
C_{\vec{y}} &= \sum_{\vec{x} \in \mathcal{A}^n} \frac{\recht{B}(\gamma+t(\vec{x}))}{\recht{B}(\gamma)}\prod_i Q_{y_i|x_i},\\
w_{\vec{x}|\vec{y}} &= C_{\vec{y}}^{-1}\frac{\recht{B}(\gamma+t(\vec{x}))}{\recht{B}(\gamma)}\prod_i Q_{y_i|x_i},\\
m_{\vec{x},\alpha} &= \frac{\gamma_{\alpha}+t_{\alpha}(\vec{x})}{\sum_{\alpha' \in \mathcal{A}} \gamma_{\alpha'}+n},\\
m_{\vec{y},\alpha} &= \sum_{\vec{x} \in \mathcal{A}^n} w_{\vec{x}|\vec{y}}m_{\vec{x},\alpha},\\
\sigma^2_{\vec{x},\alpha} &= \frac{m_{\vec{x},\alpha}(1-m_{\vec{x},\alpha})}{\sum_{\alpha' \in \mathcal{A}} \gamma_{\alpha'}+n+1}.
\end{align}
Then the $\alpha$-coefficient of $\Pi_{\recht{opt}}(\vec{y})$ is equal to $m_{\vec{y},\alpha}$, and $\recht{MSE}_{\mu}^{\recht{distr}}(\mathcal{Q},\Pi_{\recht{opt}})$ is equal to
\begin{equation}
\sum_{\vec{y} \in \mathcal{B}^n}C_{\vec{y}}\sum_{\alpha \in \mathcal{A}}  \left(\sum_{\vec{x} \in \mathcal{A}^n} w_{\vec{x}|\vec{y}}(m^2_{\vec{x},\alpha}+\sigma^2_{\vec{x},\alpha})-m_{\vec{y},\alpha}^2\right). \label{eq:mseopt}
\end{equation}
\end{proposition}

\begin{proof}
For $\gamma' \in \mathbb{R}^a_{\geq 0}$, let $\Delta_{\gamma'}(p) = \frac{1}{\recht{B}(\gamma')}\prod_{\alpha} p_{\alpha}^{\gamma'_{\alpha}}$ be the probability density function of the Dirichlet distribution with parameter vector $\alpha'$. Let $\delta_{P|\vec{Y}=\vec{y}}$ be the posterior probability density function of $P$ given $\vec{Y} = \vec{y}$. By \cite[Thm.~9.1]{lopuhaa2019information} we have (note the definition of $C_{\vec{y}}$ there differs by a factor $\recht{B}(\gamma)$ from the one given here):
\begin{equation}
\delta_{P|\vec{Y}=\vec{y}} = \frac{1}{C_{\vec{y}}} \sum_{\vec{x} \in \mathcal{A}^n} w_{\vec{x}|\vec{y}} \Delta_{\gamma+t(\vec{x})}.
\end{equation}
In other words, $P|\vec{Y}=\vec{y}$ is a mixture of Dirichlet distributions $D_{\vec{x}}$. These are parametrised by $\vec{x} \in \mathcal{A}^n$, have weight $w_{\vec{x}|\vec{y}}$, and parameter vector $\gamma+t(\vec{x})$; by Fact \ref{fact:dir} one has $\mathbb{E}[D_{\vec{x},\alpha}] = m_{\vec{x},\alpha}$ and $\recht{Var}(D_{\vec{x},\alpha})) = \sigma^2_{\vec{x},\alpha}$. Fact \ref{fact:mix} now shows us that $\Pi_{\recht{opt}}(\vec{y}) = \mathbb{E}[P_{\alpha}|\vec{Y}=\vec{y}] = m_{\vec{y},\alpha}$, proving the first claim of the proposition, and
\begin{equation}
\recht{Var}(P_{\alpha}|\vec{Y}=\vec{y}) = \sum_{\vec{x} \in \mathcal{A}^n} w_{\vec{x}|\vec{y}}(m^2_{\vec{x},\alpha}+\sigma^2_{\vec{x},\alpha})-m_{\vec{y},\alpha}^2. \label{eq:dirproof}
\end{equation}
Furthermore, note that 
\begin{align}
\mathbb{P}(\vec{X} = \vec{x}) &= \int \Delta_{\gamma}(p) \mathbb{P}(\vec{X}=\vec{x}|P=p)\textrm{d}p \\
&=  \frac{1}{\recht{B}(\gamma)}\int \prod_{\alpha} p^{\gamma_{\alpha}+t_{\alpha}(\vec{x})}\textrm{d}p \\
&= \frac{\recht{B}(\gamma+t(\vec{x}))}{\recht{B}(\gamma)}\int \Delta_{\gamma+t(\vec{x})}\textrm{d}p \\
&= \frac{\recht{B}(\gamma+t(\vec{x}))}{\recht{B}(\gamma)},\\
\mathbb{P}(\vec{Y}=\vec{y}) &= \sum_{\vec{x}} \mathbb{P}(\vec{Y}=\vec{y}|\vec{X}=\vec{x})\mathbb{P}(\vec{X}=\vec{x}) \\
&=\sum_{\vec{x}} \frac{\recht{B}(\gamma+t(\vec{x}))}{\recht{B}(\gamma)} \prod_{i=1}^n Q_{y_i|x_i} \\
&= C_{\vec{y}}.
\end{align}
Combining this with (\ref{eq:varbound}) and (\ref{eq:dirproof}), this shows that
\begin{align}
&\recht{MSE}_{\mu}^{\recht{distr}}(\mathcal{Q},\Pi_{\recht{opt}}) \nonumber \\
&= \sum_{\vec{y} \in \mathcal{B}^n} \mathbb{P}(\vec{Y}=\vec{y})\sum_{\alpha \in \mathcal{A}} \recht{Var}(P_{\alpha}|\vec{Y}=\vec{y}) \\
&= \sum_{\vec{y} \in \mathcal{B}^n}C_{\vec{y}}\sum_{\alpha \in \mathcal{A}}  \left(\sum_{\vec{x} \in \mathcal{A}^n} w_{\vec{x}|\vec{y}}(m^2_{\vec{x},\alpha}+\sigma^2_{\vec{x},\alpha})-m_{\vec{y},\alpha}^2\right). \qedhere
\end{align}
\end{proof}

Using these formulas allows us to prove Proposition \ref{prop:dirichlet}.

\begin{proof}[Proof of Proposition \ref{prop:dirichlet}]
Let $T_{\alpha}$, $S_{\beta}$ and $S_{\beta|\alpha}$ be as in (\ref{eq:deft},\ref{eq:defs},\ref{eq:defs2}). Note that $t_{\alpha} = \sum_{\beta} s_{\beta|\bullet}$, where $s_{\beta|\bullet} = (s_{\beta|1},\cdots,s_{\beta|a})$. Furthermore, for a given $s$, define $\mathcal{S}_s = \{(s_{\beta|\alpha})_{\beta,\alpha}: \forall \beta \sum_{\alpha} s_{\beta|\alpha} = s_{\beta}\}$. Then
\begin{align}
&\Pi_{\recht{opt}}(\vec{y})_{\alpha} \nonumber \\
&= \sum_{\vec{x} \in \mathcal{A}^n} w_{\vec{x}|\vec{y}}m_{\vec{x},\alpha} \\
&= C_{\vec{y}}^{-1}\sum_{\vec{x} \in \mathcal{A}^n}\frac{\gamma_{\alpha}+t_{\alpha}(\vec{x})}{\sum_{\alpha' \in \mathcal{A}} \gamma_{\alpha'}+n} \cdot \frac{\recht{B}(\gamma+t(\vec{x}))}{\recht{B}(\gamma)}\prod_i Q_{y_i|x_i} \\
&= C_{\vec{y}}^{-1}\sum_{s_{\bullet|\bullet}\in \mathcal{S}_s} \left(\prod_\beta \binom{s_{\beta}}{s_{\beta|\bullet}}\right) \frac{(\gamma_{\alpha}+t_{\alpha}(\vec{x}))\recht{B}(\gamma+t(\vec{x}))}{(\sum_{\alpha' \in \mathcal{A}} \gamma_{\alpha'}+n)\recht{B}(\gamma)} \prod_{\beta,\alpha} Q_{\beta|\alpha}^{s_{\beta|\alpha}}.
\end{align}
Here $\binom{s_{\beta}}{s_{\beta|\bullet}} = \frac{s_{\beta}!}{\prod_{\alpha} s_{\beta|\alpha}!}$ is the multinomial coefficient. Since $\#\mathcal{S}_s = \mathcal{O}(n^{(a-1)b})$, we can find $\Pi_{\recht{opt}}(\vec{y})_{\alpha}$, up to the scaling factor $C_{\vec{y}}^{-1}$, can be found by calculating $\mathcal{O}(n^{(a-1)b})$ summands. We can then find $C_{\vec{y}}$ by using the fact that $\sum_{\alpha} \Pi_{\recht{opt}}(\vec{y})_{\alpha} = 1$.

Analogous to the above, we can similarly show that we need $\mathcal{O}(n^{(a-1)b})$ to calculate each summand of the form 
\begin{equation}
\sum_{\vec{x} \in \mathcal{A}^n} w_{\vec{x}|\vec{y}}(m^2_{\vec{x},\alpha}+\sigma^2_{\vec{x},\alpha})-m_{\vec{y},\alpha}^2
\end{equation}
in (\ref{eq:mseopt}). We then need to sum over all possible $s$, of which there are $\mathcal{O}(n^{b-1})$, leading to a total complexity of $\mathcal{O}(n^{(a-1)b}\cdot n^{b-1}) = \mathcal{O}(n^{ab-1})$ to calculate $\recht{MSE}_{\mu}^{\recht{distr}}(\mathcal{Q},\Pi_{\recht{opt}})$.
\end{proof}

\section{MLE Algorithm for RR} \label{app:mle}

We describe an efficient algorithm solving (\ref{eq:minimise}) for RR with $O(a\log a)$ time complexity in Algorithm~\ref{algo:mle_RR}. This algorithm is justified as follows. Recall from \cite{warner1965randomized} that RR (with privacy parameter $\varepsilon$) is given by the $a\times a$-matrix $Q$ satisfying
\begin{equation}
    Q_{y|x} = \frac{1+(\textrm{e}^{\varepsilon}-1)\delta_{x=y}}{\textrm{e}^{\varepsilon}+a-1}.
\end{equation}
Since $\mathcal{B} = \mathcal{A}$ for RR, we can rewrite (\ref{eq:minimise}) to
\begin{align}
\recht{maximise}_p  & \ \ f(p) = \sum_{y \in \mathcal{A}} s_y \log((Qp)_y) \\
\textrm{subject to} & \ \ p \geq 0,  \sum_{x \in \mathcal{A}}p_x = 1. \nonumber
\end{align}
From this, we obtain the Karush-Kuhn-Tucker (KKT) conditions (with extra variables $(u_x)_{x \in \mathcal{A}}$ and $v$):
\begin{align*}
    \forall x \in \mathcal{A}, & \frac{\partial f(p)}{\partial p_x} + u_x - v = 0, && \text{ (stationarity)}\\
    & u_x p_x = 0,&& \text{ (complementary slackness)} \\
    & u_x \geq 0, && \text{ (dual feasibility)}\\
    & p_x \geq 0, \sum_{x \in \mathcal{A}} p_x = 1 && \text{ (primal feasibility),}
\end{align*}
in which
\begin{align}
    \frac{\partial f(p)}{\partial p_x} &= \sum_{y \in \mathcal{A}}\frac{s_y Q_{y|x}}{\sum_{k\in\mathcal{A}} Q_{y|k}p_k} \\
    & = \frac{(\textrm{e}^\varepsilon-1) s_x}{(\textrm{e}^{\varepsilon} - 1)p_x + 1} + \sum_{y \in \mathcal{A}}\frac{s_y}{(\textrm{e}^{\varepsilon}-1)p_y + 1}. 
\end{align}
By stationarity and complementary slackness we find for all $x \in \mathcal{A}$ that
\begin{align} \label{eq:v}
p_x\left(v - \frac{\partial f(p)}{\partial p_x}\right) = 0.
\end{align}
By summing all $x \in \mathcal{A}$, we find $v = n$. Suppose we have found the optimal $\hat{p}$, and define $\mathcal{A}' = \{x : \hat{p}_x > 0\}$. 
For $x,x' \in \mathcal{A}'$, it follows from (\ref{eq:v}) that we have
\begin{align}
    \frac{ s_x}{(\textrm{e}^{\varepsilon} - 1)\hat{p}_x + 1} &= \frac{ s_{x'}}{(\textrm{e}^{\varepsilon} - 1)\hat{p}_{x'} + 1},
\end{align}
hence
\begin{align}
\frac{s_{x'}}{s_x}((\textrm{e}^{\varepsilon} - 1)\hat{p}_x+1) = (\textrm{e}^{\varepsilon} - 1)\hat{p}_{x'}+1.
\end{align}
Summing over all $x' \in \mathcal{A}'$, we can solve $\hat{p}_i$ as
\begin{align} \label{eq:hatp}
    \hat{p}_x = \frac{\#\mathcal{A}' - \frac{\sum_{x'\in \mathcal{A}'}s_{x'}}{s_x}  + (\textrm{e}^{\varepsilon} - 1)}{(\textrm{e}^{\varepsilon} - 1)\left(\frac{\sum_{x'\in \mathcal{A}'}s_{x'}}{s_c} \right)}.
\end{align}
Therefore, the problem of finding $\hat{p}$ reduces to finding $\mathcal{A}'$. We determine $\mathcal{A}'$ by starting with $\mathcal{A}' = \mathcal{A}$, and then repeatedly removing the $x$ with the lowest value of $s_x$, until the estimation (\ref{eq:hatp}) is no longer nonnegative for all $x$.

\begin{algorithm}
\SetAlgoLined
\SetKwInOut{Input}{Input}\SetKwInOut{Output}{Output}
\SetKwFunction{Sort}{Sort}
\Input{Total number of users $n$; privacy budget $\varepsilon$; tallies of obfuscated data $s=[s_0, \ldots, s_{a-1}]$}
\Output{MLE estimate distribution $\hat{p} = [\hat{p}_1, \ldots, \hat{p}_{a-1}]$}
\BlankLine
$[s'_{(0)}, \ldots, s'_{(a-1)} ]\leftarrow$ \Sort{s}, such that $\forall i < j, s'_{(i)} < s'_{(j)}$ \;
$k = 0$ \;
 \While{$a-k  + \textrm{e}^{\varepsilon} - 1 - \frac{\sum_{i=k}^{a-1}s'_{(i)}}{s'_{(k)}}< 0$}{
  $k \leftarrow k + 1$ \;
 }
 \For{$j = 0, \ldots, a-1$}{
 $\Phi_j = \frac{s_i}{\sum_{i=k}^{a-1}s'_{(i)} }$ \;
 $\hat{p}_j = \max\{0,\quad \Phi_j \left(\frac{n-k}{\textrm{e}^{\varepsilon} - 1}+1 \right) - \frac{1}{\textrm{e}^{\varepsilon} - 1}\}$
 }
 \caption{Exact MLE post-processing for RR}
 \label{algo:mle_RR}
\end{algorithm}

\end{document}